%% file: main.tex
\documentclass{article}

\usepackage[affil-it]{authblk}
\usepackage[dvipsnames]{xcolor}
\usepackage{amsfonts}
\usepackage{amsmath,amsthm,amssymb,dsfont}

\usepackage{enumerate}
\usepackage{graphicx}	
\usepackage{subcaption}
\usepackage[margin=3cm]{geometry}
\usepackage{url}
\usepackage{todonotes}
\usepackage{bbm}

\usepackage{tikz}

\usepackage{pifont}
\usepackage{multirow}
\usepackage{makecell}

\usepackage{epsfig}
\usetikzlibrary{shapes.symbols,patterns} 
\usepackage{pgfplots}
\pgfplotsset{compat=1.10}
\usepgfplotslibrary{fillbetween}

\definecolor{linkblue}{HTML}{001487}
\usepackage{hyperref}
\hypersetup{colorlinks=true,citecolor=linkblue,linkcolor=linkblue,filecolor=linkblue,urlcolor=linkblue,breaklinks=true}

\usepackage{nicefrac}
\usepackage{mathtools}

\usepackage{thmtools} 
\hypersetup{hypertexnames=false}

\usepackage{algorithm}
\usepackage{algorithmic}

\usepackage{mdframed}
\usepackage{aligned-overset}
\usepackage{circuitikz}

\theoremstyle{plain}
\newtheorem{theorem}{Theorem}[section]
\newtheorem{lemma}[theorem]{Lemma}

\newtheorem{corollary}[theorem]{Corollary}
\newtheorem{proposition}[theorem]{Proposition}

\theoremstyle{definition}
\newtheorem{definition}[theorem]{Definition}
\newtheorem{remark}[theorem]{Remark}
\newtheorem{example}[theorem]{Example}

\usepackage{tcolorbox}
\tcolorboxenvironment{definition}{}
\tcolorboxenvironment{theorem}{colback=pink!25!white,colframe=pink!100!black}
\tcolorboxenvironment{lemma}{colback=yellow!5!white,colframe=yellow!75!black}
\tcolorboxenvironment{corollary}{colback=Dandelion!5!white,colframe=Dandelion!75!black}
\tcolorboxenvironment{proposition}{colback=Emerald!5!white,colframe=Emerald!75!black}
\tcolorboxenvironment{claim}{colback=RoyalBlue!5!white,colframe=RoyalBlue!75!black}
\tcolorboxenvironment{conjecture}{colback=red!5!white,colframe=red!75!black}
\tcolorboxenvironment{example}{colback=white,colframe=lightgray}
\tcolorboxenvironment{remark}{colback=white,colframe=lightgray}
\tcolorboxenvironment{question}{colback=lightgray!5!white,colframe=lightgray!75!black}

\DeclareRobustCommand{\abbrevcrefs}{%
\Crefname{theorem}{Thm.}{Thms.}%
\Crefname{corollary}{Cor.}{Cors.}%
\Crefname{lemma}{Lem.}{Lems.}%
\Crefname{remark}{Rmk.}{Rmks.}%
\Crefname{proposition}{Prop.}{Props.}%
\Crefname{equation}{Eq.}{Eqs.}%
\Crefname{example}{Ex.}{Exs.}%
}

\DeclareMathOperator*{\argmin}{argmin}

\DeclareRobustCommand{\Cshref}[1]{{\abbrevcrefs\Cref{#1}}}

\newcommand*{\ee}{\mathrm{e}}

\newcommand*{\cA}{\mathcal{A}}
\newcommand*{\cB}{\mathcal{B}}

\newcommand*{\cE}{\mathcal{E}}
\newcommand*{\cF}{\mathcal{F}}

\newcommand*{\cH}{\mathcal{H}}

\newcommand*{\cI}{\mathcal{I}}
\newcommand*{\cK}{\mathcal{K}}

\newcommand*{\cM}{\mathcal{M}}
\newcommand*{\cP}{\mathcal{P}}

\newcommand*{\cR}{\mathcal{R}}
\newcommand*{\cS}{\mathcal{S}}
\newcommand*{\cT}{\mathcal{T}}

\newcommand*{\cV}{\mathcal{V}}

\newcommand*{\cX}{\mathcal{X}}

\newcommand*{\N}{\mathbb{N}}

\newcommand*{\R}{\mathbb{R}}

\newcommand*{\C}{\mathbb{C}}

\newcommand*{\St}{\mathrm{S}}

\newcommand*{\reg}{\infty}

\newcommand*{\eps}{\varepsilon}

\newcommand*{\id}{\mathds{1}}
\newcommand*{\poly}{\mathrm{poly}}

\newcommand*{\supp}{\mathrm{supp}}
\newcommand*{\tr}{\mathrm{tr}}
\newcommand*{\ket}[1]{| #1 \rangle}
\newcommand*{\bra}[1]{\langle #1 |}
\newcommand*{\spr}[2]{\langle #1 | #2 \rangle}
\newcommand{\proj}[1]{|#1\rangle\!\langle #1|}
\newcommand*{\braket}[1]{\langle #1 \rangle}

\newcommand*{\D}{\mathrm{D}}
\newcommand*{\conv}{\mathrm{conv}}

\newcommand*{\ci}{\mathrm{i}} 
\newcommand*{\di}{\mathrm{d}} 

\newcommand{\norm}[1]{\left\lVert#1\right\rVert}

\usepackage{float}
\usepackage[nameinlink,capitalize,noabbrev]{cleveref}

 \allowdisplaybreaks

\title{Robust generalized quantum Stein's lemma}

\author{\normalsize Giulia Mazzola$^{1}$, David Sutter$^{2}$, and Renato Renner$^{1}$}
 \affil{\small $^{1}$Institute for Theoretical Physics, ETH Zurich\\
 $^{2}$IBM Research Europe -- Zurich
}
 \date{}

\begin{document}

\maketitle

\begin{abstract}
The generalized quantum Stein's lemma provides an explicit expression for the optimal error exponent when distinguishing many independent and identically distributed (iid) copies of a given bipartite state from the set of separable bipartite states. Here we prove that this result is robust, in the sense that the iid assumption can be relaxed to almost-iid.  In particular, our result shows that the original argument of Brand\~ao and Plenio, which contains a logical gap, can be made rigorous. Our proof relies on a novel continuity bound for the relative entropy of entanglement with respect to the quantum Wasserstein distance. Combined with a recent insight that almost-iid states and their exact iid counterparts are asymptotically close in this distance, the bound implies that their relative entropies of entanglement coincide asymptotically.
\end{abstract}

\section{Introduction}
Suppose we are given an unknown source that produces either a quantum state $\rho$ (null hypothesis) or $\sigma$ (alternative hypothesis). The task is to perform an optimal measurement that allows us to determine whether the source produced $\rho$ or $\sigma$. Since the source is binary, the optimal measurement can be described by a two-outcome positive operator-valued measure (POVM) $\{M,\id - M\}$.
There are two types of errors that can occur: we guess $\sigma$ but the source produced $\rho$ (type I error), which occurs with probability $\alpha(M,\rho):=\tr[\rho (\id-M)]$, or we guess $\rho$ but the source produced $\sigma$ (type II error), which occurs with probability $\beta(M,\sigma):=\tr[\sigma M]$. 
\emph{Asymmetric hypothesis testing} aims to minimize the probability of a type II error for a fixed type I error probability of at most $\eps \in (0,1)$, which is captured by the quantity 
\begin{align} \label{eq_beta_function}
 \beta_{\eps}(\rho\|\sigma) := \min_{0 \leq M \leq \id} \big\{ \beta(M,\sigma): \alpha(M,\rho) \leq \eps \big\}  \, .
\end{align}

In the simplest setting, we assume that the source outputs independent and identically distributed (iid) copies of the states (see~\cref{fig_iid}). We then want to understand how fast the error probability $\beta_{\eps}(\rho^{\otimes n}\|\sigma^{\otimes n})$ decays in $n$, which can be rewritten as an error exponent $-\frac{1}{n} \log \beta_{\eps}(\rho^{\otimes n}\|\sigma^{\otimes n}) = \frac{1}{n} D^{\eps}_H(\rho^{\otimes n}\|\sigma^{\otimes n})$, where $D^\eps_H$ denotes the \emph{hypothesis testing relative entropy}~\cite{buscemi10,WR12}
 \begin{align} \label{eq_def_DH}
 D_H^\eps(\rho \| \sigma) := - \log \min_{0 \leq X \leq \id } \big\{\tr[\sigma X]:  \tr[\rho X] \geq 1- \eps \big\} \, .
 \end{align}
The \emph{quantum Stein's lemma}~\cite{PH91,ogawa00} proves that the error exponent for asymmetric hypothesis testing in an iid setting is given by the \emph{relative entropy}, i.e.
\begin{align} \label{eq_SL}
\lim_{n \to \infty} \frac{1}{n} D^{\eps}_H(\rho^{\otimes n} \| \sigma^{\otimes n}) = D(\rho \| \sigma) \, ,
\end{align}
where $D(\rho \| \sigma):=\tr[\rho (\log \rho - \log \sigma)]$ if the support of $\rho$ is included in the support of $\sigma$, and $+\infty$ otherwise.
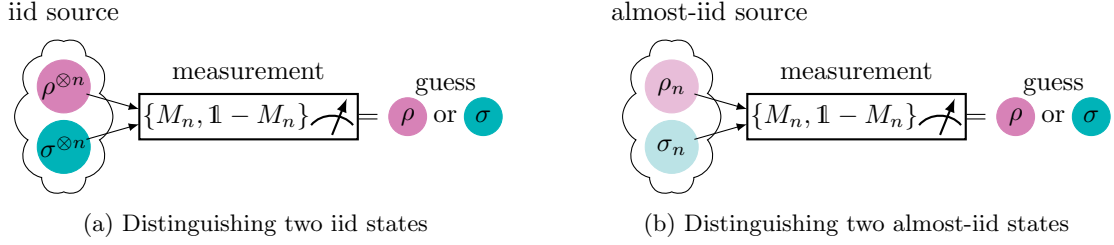
\begin{figure}[!htb]
    \centering
    \begin{subfigure}[t]{0.49\textwidth}
        \centering
         \input{iid_HT}
        \caption{Distinguishing two iid states\\ \phantom{b}}
        \label{fig_iid}
    \end{subfigure}
        \hfill
    \begin{subfigure}[t]{0.49\textwidth}
            \centering
        \input{non_iid_HT}
        \caption{Distinguishing two almost-iid states}
        \label{fig_non_iid}
    \end{subfigure}
\caption{Hypothesis testing in an iid and an almost-iid setting. In~\cref{fig_iid} we have a source that outputs a state that is iid in $\rho$ or $\sigma$. In~\cref{fig_non_iid} the source generates a state that is almost-iid along $\rho$ or $\sigma$. After an optimized measurement, described by a POVM $\{M_n,\id - M_n\}$, the task is to guess whether the source produced a state corresponding to $\rho$ or $\sigma$. 
}
\label{fig_HT}
\end{figure}

The assumption of a perfect iid structure is hard to justify in practice. Therefore, it is desirable to relax this assumption, which leads to the notion of almost-iid states.  Concretely, we consider a superposition of iid states that contain a small number of defects, denoted by $r$, which usually scales sublinearly in $n$. We say $\rho_n$ is a $\binom{n}{r}$-almost-iid state along $\rho$ and denote the set of all such states by $\St^n(\cH, \rho^{\otimes n-r})$.\footnote{See~\cref{def_almost_product_state_mixed} for a precise definition of almost-iid states.} 
We refer to~\cite{MSR26} for a more detailed discussion on the physical relevance of almost-iid states and their mathematical properties.

This motivates the study of a robust hypothesis testing setting, where we do not want to distinguish between $\rho^{\otimes n}$ and $\sigma^{\otimes n}$, but rather between the two sets of \smash{$\binom{n}{r_1}$}- and \smash{$\binom{n}{r_2}$}-almost-iid states along $\rho$ and $\sigma$, respectively (see~\cref{fig_non_iid}).
Such a setting generalizes the iid case mentioned in~\cref{eq_SL}, because the iid state corresponds to an almost-iid set with zero defects ($r=0$), i.e.~$\{\rho^{\otimes n}\} = \St^n(\cH, \rho^{\otimes n})$.
We quantify the optimal error exponent for this task, which measures the ``worst-case error'' in the sense that it upper bounds the error for any possible almost-iid states along $\rho$ and $\sigma$, with at most $r_1$ and $r_2$ defects, respectively. 
In mathematical terms, we show that for $r_1=o(n)$ and $r_2=o(n)$
\begin{align} \label{eq_ASL}
\lim_{n \to \infty} \frac{1}{n} \min_{\rho_n \in \St^n(\cH, \rho^{\otimes n-r_1})} \min_{\sigma_n \in \St^n(\cH, \sigma^{\otimes n-r_2})}  D^{\eps}_H(\rho_n \| \sigma_n) = D(\rho \| \sigma) \, .
\end{align}  
We refer to \Cref{eq_ASL} as a \emph{robust quantum Stein’s lemma}: it states that, for almost-iid states, the optimal error exponent coincides with that of perfect iid states in the original quantum Stein’s lemma. The latter is recovered as a special case of the robust version by setting the numbers of defects to $r_1=r_2=0$, so that \cref{eq_ASL} reduces to \cref{eq_SL}. The precise statement is given in \cref{thm_ASL}.

In~\cite{brandao_Stein_10}, the setting has been generalized such that the alternative hypothesis consists of a family of states, which does not need to have an iid or almost-iid structure. For example, we may consider a bipartite Hilbert space $\cH_{AB} = \cH_A \otimes \cH_B$ and density matrices $\rho_{AB}$ and $\sigma_{AB}$ where $\sigma_{AB}$ is from the set of separable density operators, which is defined as $\mathrm{SEP}(A:B):=\conv\{ \proj{\phi}_A \otimes \proj{\varphi}_B : \ket{\phi}_A \in \cH_A, \ket{\varphi}_B \in \cH_B, \braket{\phi|\phi}=\braket{\varphi|\varphi}=1\}$. We then want to distinguish between $\rho_{AB}^{\otimes n}$ and any arbitrary $\sigma_n \in \mathrm{SEP}(A^n:B^n)$, which need not be an iid or an almost-iid state.
Similarly to before, the goal is to understand how quickly $\max_{\sigma_n \in \mathrm{SEP}(A^n:B^n)} \beta_{\eps}(\rho^{\otimes n}\|\sigma_n)$ decays to zero as $n$ grows large, or equivalently, the asymptotic behavior of the error exponent \smash{$\frac{1}{n} \min_{\sigma_n \in \mathrm{SEP}(A^n:B^n)} D^{\eps}_H(\rho^{\otimes n} \| \sigma_n)$}. This setting is sometimes referred to as \emph{entanglement testing} (see~\cref{fig_iid_entag}).

A popular measure to quantify the amount of entanglement is the \emph{relative entropy of entanglement}~\cite{VPRK97} defined as
\begin{align}  \label{eq_def_RE}
D(\rho_{AB} \| \mathrm{SEP}):= \min_{\sigma_{AB} \in \mathrm{SEP}(A:B)} D(\rho_{AB} \| \sigma_{AB}) \, .
\end{align}
The relative entropy of entanglement is not additive under the tensor product~\cite{VW01}, which justifies the definition of a regularized version $D^{\reg}(\rho_{AB} \| \mathrm{SEP}):=\lim_{n \to \infty} \frac{1}{n} D(\rho_{AB}^{\otimes n} \| \mathrm{SEP}(A^n:B^n))$. The limit in the regularization exists due to Fekete's subadditivity lemma, as explained in~\cref{lem_fekete}.
\begin{figure}[!htb]
    \centering
    \begin{subfigure}[b]{0.45\textwidth}
        \centering
         \input{iid_HT_ET}
        \caption{Distinguishing an iid from a separable state}
        \label{fig_iid_entag}
    \end{subfigure}
        \hfill
    \begin{subfigure}[b]{0.53\textwidth}
            \centering
        \input{non_iid_HT_ET}
        \caption{Distinguishing an almost-iid from a separable state}
        \label{fig_non_iid_entag}
    \end{subfigure}
\caption{Entanglement testing in an iid and an almost-iid setting. In~\cref{fig_iid_entag} we have a source that outputs an iid state in $\rho_{AB}$ or a separable state. In~\cref{fig_non_iid_entag} the source generates a state that is almost-iid along $\rho_{AB}$ or a separable state.
After an optimized measurement, described by a POVM $\{M_n,\id - M_n\}$, the task is to guess if the source produced a state corresponding to $\rho_{AB}$ or a separable state.}
\label{fig_HT_entag}
\end{figure}
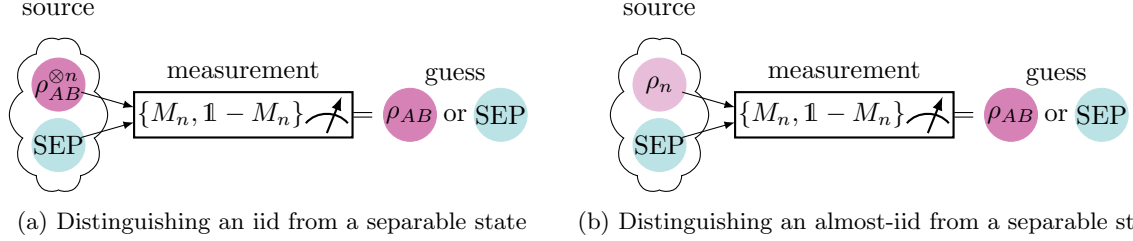

The \emph{generalized quantum Stein's lemma}~\cite{brandao_Stein_10,haya_stein_25,ludo25} states that the error exponent for entanglement testing is given by the regularized relative entropy of entanglement, i.e.,
\begin{align} \label{eq_generalized_quantum_stein_intro}
\lim_{n \to \infty} \frac{1}{n} \min_{\sigma_n \in \mathrm{SEP}(A^n:B^n)} D^{\eps}_H(\rho_{AB}^{\otimes n} \| \sigma_n) 
= D^{\reg}(\rho_{AB} \| \mathrm{SEP}) \, .
\end{align}
Analogously to the robust version of the quantum Stein's lemma discussed above, we also prove a robust version of this generalization, i.e., we quantify the error exponent for entanglement testing of almost-iid states --- again in a ``worst-case'' setting in the sense that we upper bound the error for any almost-iid state with $r$ defects (see~\cref{fig_non_iid_entag}). In mathematical terms, we show that for $r=o(n)$
\begin{align} \label{eq_AGSL}
\lim_{n \to \infty} \frac{1}{n} \min_{\rho_n \in \St^n(\cH_{AB}, \rho_{AB}^{\otimes n-r})}  \min_{\sigma_n \in \mathrm{SEP}(A^n:B^n)} D^{\eps}_H(\rho_n \| \sigma_n) 
=  D^{\reg}(\rho_{AB} \| \mathrm{SEP}) \, .
\end{align}
We call \Cref{eq_AGSL} a \emph{robust generalized quantum Stein's lemma} and refer to~\cref{thm_AGSL} for the precise statement. Choosing $r=0$ reveals that~\cref{eq_AGSL} implies~\cref{eq_generalized_quantum_stein_intro}. We note that special cases of \cref{eq_AGSL} have been known before. In particular, the case where the defect parameter $r$ is constant in $n$ was established in~\cite[Theorem~32]{ludo25}. Furthermore, the case where all states correspond to classical probability distributions while $r=o(n)$ was proved in~\cite[Section~5.3]{Ludo25_2}.	
Recently, the optimal error exponent was determined for a related setting, where the null hypothesis is composed of iid copies of an unknown state~\cite{LRT26}.

To prove our results, we leverage the mathematical structure of almost-iid states. The quantum Wasserstein distance introduced in~\cite{PMTL21} provides a useful tool for their characterization, as it is inherently robust under local changes. Recently, it was observed in~\cite[Corollary 18]{Giacomo_private} that almost-iid states (see~\cref{def_almost_product_state_mixed}) are close in quantum Wasserstein distance to their corresponding iid states. This naturally raises the question of whether the relative entropy of entanglement is continuous with respect to the quantum Wasserstein distance. We answer this question in the affirmative. Combined with the insight from~\cite{Giacomo_private}, this allows us to derive~\cref{eq_AGSL}.

\paragraph{Main results} The main contributions of this paper are summarized as follows:
\begin{enumerate}[(a)]
\item We prove a robust generalized quantum Stein's lemma (see~\cref{eq_AGSL}), which shows that the optimal error exponent for entanglement testing with respect to almost-iid states is given by the regularized relative entropy of entanglement of the corresponding iid state. We refer to~\cref{thm_AGSL} for a rigorous statement. \label{res_AGSL}
\item Based on~\eqref{res_AGSL}, we prove a robust version of the quantum Stein's lemma (see~\cref{eq_ASL}), which shows that the relative entropy is not only the optimal error exponent between distinguishing between two iid sources but also for two almost-iid sources. We refer to~\cref{thm_ASL} for a rigorous statement. 
\item We prove an explicit continuity bound for the relative entropy of entanglement with respect to the quantum Wasserstein distance~\cite{PMTL21} (see~\cref{prop_continuity_ER_Wasserstein}). Since the Wasserstein distance is weaker than the trace distance, this bound is stronger than previously known continuity bounds formulated in terms of the trace distance. 
Moreover, as almost-iid states are close to iid states with respect to the Wasserstein distance~\cite{Giacomo_private}, the relative entropy of entanglement asymptotically coincides for almost-iid and iid states (see~\cref{cor_asymptotic_continuity_ER}).
\item  Our results show that the original proof of the generalized quantum Stein's lemma~\cite{brandao_Stein_10}, which contains a lemma with a logical gap~\cite{Berta2023gapinproofof,stein_nature_24}, is correct. This yields a third independent proof (based on the exponential de Finetti theorem and properties of almost-iid states) for the generalized quantum Stein's lemma after a proof~\cite{haya_stein_25} (based on the spectral pinching method) and another proof~\cite{ludo25} (based on the quantum blurring technique).
\end{enumerate}
\Cref{fig_results} shows how the new results relate to the existing (generalized) quantum Stein's lemma.
The technical results given in~\cref{sec_robustness_ER,sec_robust_GSL,sec_superadditivity} are stated in a general resource-theoretic framework as done in~\cite{brandao_Stein_10,haya_stein_25,ludo25}, which goes beyond entanglement testing. 
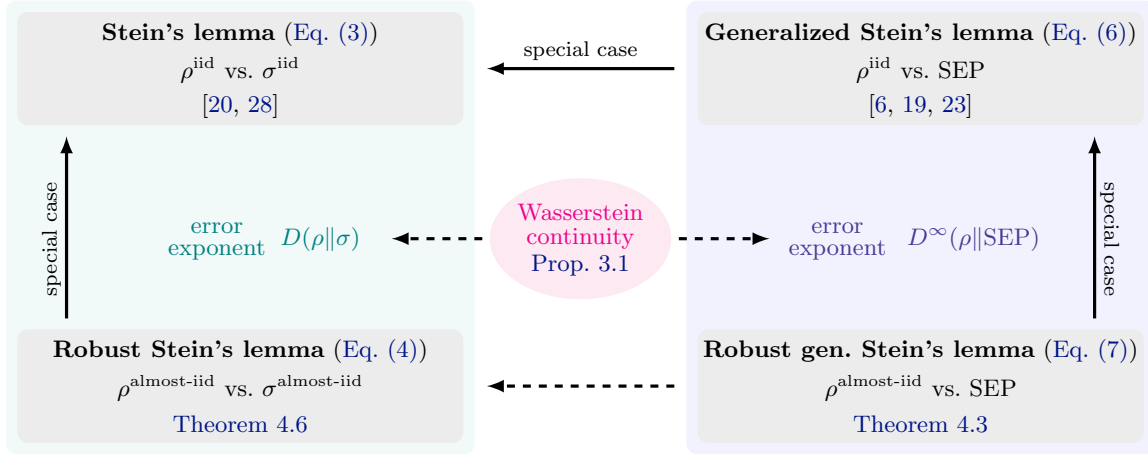
\begin{figure}[!htb]
    \centering
\input{overview_results}
\caption{Overview of the results. Solid arrows denote straightforward implications, while dashed arrows are nontrivial ones.}
\label{fig_results}
\end{figure}

\paragraph{Paper organization}
\Cref{sec_intro} reviews the precise definition of almost-iid states and of relevant entropic quantities. In~\cref{sec_robustness_ER} we show that the relative entropy of entanglement is continuous with respect to the Wasserstein distance, implying that it exhibits the same asymptotic behavior for almost-iid and iid states. This technical result serves as a key ingredient in the proof of the robust (generalized) quantum Stein's lemma, presented in~\cref{sec_robust_GSL}. We conclude in~\cref{sec_superadditivity} showing that the relative entropy of entanglement is asymptotically superadditive for states with equal marginals. This result may be of independent interest and follows from the novel continuity statement presented in~\cref{sec_robustness_ER}. 

\section{Preliminaries} \label{sec_intro}
\subsection{Almost-iid states} \label{subsec_almost_iid}
Let $\St(\cH)$ denote the set of density matrices on $\cH$, $\cS_n$ be the set of permutations on $[n]:=\{1,\ldots, n\}$, and define the set $\cV(\cH^{\otimes n}, \ket{\theta}^{\otimes m}):=\{\pi(\ket{\theta}^{\otimes m} \otimes \ket{\Omega^{(n -m)}}): \pi \in \cS_n, \ket{\Omega^{(n -m)}} \in \cH^{\otimes n -m} \}$.
For an arbitrary matrix $X$ we denote its trace-norm by $\norm{X}_1 := \tr[\sqrt{X^\dagger X}]$. 

\begin{definition}[Almost-iid states~\cite{MSR26}] \label{def_almost_product_state_mixed}
Let $\cH_A$ be a Hilbert space, $\sigma_A \in \St(\cH_A)$, and $n,r \in \N$ such that $r \leq n$. Then, $\rho_{A^n} \in \St(\cH_A^{\otimes n})$ is called a $\binom{n}{r}$-\emph{almost-iid state along} $\sigma_A $ if there exists a purification $\ket{\theta}_{AE}$ of $\sigma_A$ and an extension $\rho_{A^n E^n}$ of $\rho_{A^n}$ such that
\begin{enumerate}[(i)]
\item $\rho_{A^n E^n}$ is permutation invariant under any transition $(A_i,E_i) \leftrightarrow (A_j,E_j)$;\label{it_first_def}
\item $\supp(\rho_{A^n E^n}) \subseteq  \mathrm{span}\,\cV(\cH_{AE}^{\otimes n},\ket{\theta}_{AE}^{\otimes n-r})$. \label{it_second_def}
\end{enumerate}
The set of $\binom{n}{r}$-almost-iid states along $\sigma_A$ is denoted by $\St^n(\cH_A, \sigma_A^{\otimes n-r})$.
\end{definition}
We call $\rho_{A^n}$ a \smash{$\binom{n}{r}$}-\emph{generalized-almost-iid-state along} $\sigma_A$ if it satisfies~\eqref{it_second_def} but not necessarily~\eqref{it_first_def}. We then write  \smash{$\rho_{A^n} \in \bar \St^n(\cH_A, \sigma_A^{\otimes n-r})$}. Trivially, we have $\St^n(\cH_A, \sigma_A^{\otimes n-r}) \subseteq \bar \St^n(\cH_A, \sigma_A^{\otimes n-r})$.
We refer the interested reader to~\cite{MSR26} for a detailed discussion about the relevance and mathematical properties of almost-iid states.

Almost-iid states are not necessarily close to iid states in trace distance. A simple example that demonstrates this is $\rho_{A^n}:= \sigma_{A}^{\otimes n-1} \otimes \omega_A$, which is a (generalized) $\binom{n}{1}$-almost-iid state along $\sigma_A$ for some arbitrary ``defect" term $\omega_A \in \St(\cH_A)$. However, $\| \rho_{A^n} - \sigma_A^{\otimes n}\|_1 = \| \sigma_{A}^{\otimes n-1} \otimes (\omega_A - \sigma_A)\|_1 = \norm{\omega_A - \sigma_A}_1$, which can be large.
Interestingly, almost-iid states are (asymptotically) close to iid states in a weaker distance called \emph{quantum Wasserstein distance of order 1}~\cite{PMTL21}. This is made rigorous in~\cref{lem_amost_iid_close_to_iid_W1}. 

Informally~\cite[Definition~4]{PMTL21}, two quantum states on $n$ subsystems are called neighboring if they coincide after a suitable subsystem is discarded. The quantum Wasserstein distance of order 1 is the maximum distance induced by a norm that assigns a distance at most one to any couple of neighboring states.

\begin{definition}[Quantum Wasserstein distance of order 1~\cite{PMTL21}] \label{def_wasserstein}
Let $\omega, \tau \in \St(\cH^{\otimes n})$. Then, 
\begin{align*}
\norm{\omega - \tau}_{W_1} : = \frac{1}{2} \min\bigg\{  \sum_{i=1}^n \norm{X^{(i)}}_1 :\, \tr_i[X^{(i)}] =0 ,\ \omega-\tau = \sum_{i=1}^n X^{(i)} \bigg\} \, .
\end{align*}
\end{definition}
\noindent
Since \smash{$\norm{\omega - \tau}_{W_1} \in [0,n]$}~\cite[Proposition~2]{PMTL21}, it is natural to consider the normalized Wasserstein distance $\frac{1}{n}\norm{\omega - \tau}_{W_1}$.
The normalized Wasserstein distance is a weaker metric than the trace distance in the sense that for any $\omega, \tau \in \St(\cH^{\otimes n})$ we have
\begin{align} \label{eq_W1_to_TraceDist}
\frac{1}{n} \norm{\omega - \tau}_{W_1} 
\overset{\textnormal{\cite[Prop.~2]{PMTL21}}}{\leq} \frac{1}{2} \norm{\omega - \tau}_{1} \, . 
\end{align}
It is well-known that there exist states which are close in (normalized) Wasserstein distance but not close in trace distance.
A simple example is the one mentioned above, namely $\rho_{A^n}:= \sigma_{A}^{\otimes n-1} \otimes \omega_A$. In this case $\norm{\rho_{A^n} - \sigma_A^{\otimes n}}_1 = \norm{\omega_A - \sigma_A}_1$, which can be large. At the same time~\cite[Propositions~2 and~4]{PMTL21} show that $\frac{1}{n}\norm{\rho_{A^n} - \sigma_A^{\otimes n}}_{W_1} = \frac{1}{n} \norm{\omega_A - \sigma_A}_{W_1}= \frac{1}{2n} \norm{\omega_A - \sigma_A}_1$ which is small for large $n$.

 It has first been realized by~\cite{Giacomo_private} that almost-iid and iid states are close in (normalized) Wasserstein distance. The next lemma is a variant of~\cite[Corollary~18]{Giacomo_private} that avoids the dimension-dependence and has a better asymptotic scaling. Its proof uses techniques developed in~\cite[Proposition~2.6]{MSR26} and~\cite{PMTL21}.
\begin{lemma}[Almost-iid states are close to iid states in terms of Wasserstein distance] \label{lem_amost_iid_close_to_iid_W1}
Let $n,r \in \N$ be such that $r \leq n$, $\sigma_{A} \in \St(\cH_{A})$, and $\rho_{A^n} \in \bar \St^n(\cH_{A}, \sigma_{A}^{\otimes n -r})$. Then
\begin{align}
\frac{1}{n} \norm{\rho_{A^n} - \sigma_A^{\otimes n}}_{W_1} \leq  2 \sqrt{\frac{r}{n}} \, .
\end{align}
\end{lemma}
\begin{proof}
The argument follows a strategy similar to that of~\cite[Proposition~2.6]{MSR26}.
Let $\ket{\theta}_{AE}$ be a purification of $\sigma_A$.
By~\cref{def_almost_product_state_mixed} there exists an extension $\rho_{A^{n} E^n}$ of $\rho_{A^n}$ that can be written as
\begin{align} \label{eq_ONB_dec}
\rho_{A^{n} E^n }= \sum_{t,t' \in \cT} \beta_{t,t'} \ket{ \Psi_t} \bra{\Psi_{t'}} \, ,
\end{align}
for a family $\{\ket{\Psi_t} \}_{t \in \cT}$ of orthonormal vectors from $\cV(\cH_{AE}^{\otimes n},\ket{\theta}_{AE}^{\otimes n-r})$ and  $\sum_{t \in \cT } \beta_{t,t} =1$.
For any $t \in \cT$ let $\cF_t \subseteq [n]$ denote the set of sites on which $\ket{\Psi_t}$ is not of the form $\ket{\theta}_{AE}$ and note that
\begin{align} 
| \cF_t | \leq r \, . \label{eq_size_set_of_defects}    
\end{align}
For any site $i \in [n]$, the total weight of the vectors $\ket{\Psi_t}$ that deviate from $\ket{\theta}_{AE}$ is given by
\begin{align}
w_i := \sum_{t \in \cT} \beta_{t,t} \id_{\{i \in \cF_t\}} \,, \label{eq_def_weight}
\end{align}
where $\id_{\{\textnormal{statement} \}}= 1$ if statement is true and $0$, otherwise.
Summing over all sites yields
\begin{align}
\sum_{i=1}^n w_i 
\overset{\textnormal{\Cshref{eq_def_weight}}}{=} \sum_{i=1}^n\sum_{t \in \cT} \beta_{t,t} \id_{\{ i \in \cF_t \}}  
= \sum_{t \in \cT}  \beta_{t,t}  \sum_{i=1}^n \id_{ \{ i \in \cF_t \} }  
\overset{\textnormal{\Cshref{eq_size_set_of_defects}}}{\leq}  \sum_{t \in \cT}  \beta_{t,t}  r 
=r \, . \label{eq_sum_weights}
\end{align}

For any site $i \in [n]$, let $\Pi_i$ be the projector onto the subspace spanned by $\{\ket{ \Psi_t} : i \not \in \cF_t \}$, i.e. 
\begin{align} \label{eq_def_PI}
\Pi_i = \sum_{t \in \cT : i \not \in \cF_t } \proj{\Psi_t}   \qquad \textnormal{and} \qquad \Pi_i^\perp = \sum_{t \in \cT : i \in \cF_t }\proj{\Psi_t} \, .
\end{align}
Note that $\Pi_i \, \Pi_i^\perp = 0$ and $\Pi_i + \Pi_i^\perp =\id_{\mathrm{span}\, \cV(\cH_{AE}^{\otimes n},\ket{\theta}_{AE}^{\otimes n-r})}$. 
We next observe that
\begin{align}
\tr[\Pi_i \rho_{A^n E^n} ]
\!=\!1\!-\!\tr[\Pi_i^\perp \rho_{A^n E^n} ]
\overset{\textnormal{\Cshref{eq_def_PI}}}{=} 1 \!-\!\!\!\!\!\!\! \sum_{t \in \cT: i \in \cF_t}\!\! \bra{\Psi_t} \rho_{A^n E^n} \ket{\Psi_t}
\overset{\textnormal{\Cshref{eq_ONB_dec}}}{=} 1\!-\!\!\!\!\!\!\! \sum_{t \in \cT :i \in \cF_t}\!\! \beta_{t,t}
\!\overset{\textnormal{\Cshref{eq_def_weight}}}{=}\!1 \!-\! w_i \, .  \label{eq_renato_dist4}
\end{align}
For
\begin{align} \label{eq_bar_rho}
\bar \rho^{(i)}_{A^n E^n} := \frac{\Pi_i \rho_{A^n E^n} \Pi_i }{\tr[\Pi_i \rho_{A^n E^n}]} 
\end{align}
the ``gentle measurement lemma"~\cite{winter99} implies that
\begin{align}
\norm{\bar \rho^{(i)}_{A^n E^n}  - \rho_{A^n E^n}}_1 
\overset{\textnormal{\cite[Lem.~9.4.1]{wilde_book} \& \Cshref{eq_renato_dist4,eq_bar_rho}}}{\leq} 2 \sqrt{w_i} \, . \label{eq_gentle_meas_lemma}
\end{align}
Due to the structure of $\rho_{A^n E^n}$ and the projector $\Pi_i$ we have
\begin{align} \label{eq_product_structure_after_projection}
\bar \rho^{(i)}_{A^n E^n}  = \proj{\theta}_{A_i E_i} \otimes \eta^{(i)}_{\bar A_i \bar E_i} \, ,
\end{align}
for some  $\eta^{(i)}_{\bar A_i \bar E_i}  \in \St(\cH_{AE}^{\otimes n-1})$ and  using the notation $\bar A_i := (A_1, \ldots, A_{i-1}, A_{i+1}, \ldots A_n)$.
For any $i \in [n]$, we define the density matrix
\begin{align} \label{eq_def_tau}
\tau^{(i)}_{A^n E^n}: = \rho_{A^i E^i} \otimes \proj{\theta}_{AE}^{\otimes n-i} \, ,
\end{align}
with the convention $\tau^{(0)}_{A^n E^n} = \proj{\theta}_{AE}^{\otimes n}$. Performing a telescoping sum yields
\begin{align} \label{eq_telescoping_sum}
\rho_{A^n E^n} - \proj{\theta}_{AE}^{\otimes n} = \sum_{i=1}^n \big(\tau^{(i)}_{A^n E^n} - \tau^{(i-1)}_{A^n E^n} \big) 
\end{align} 
and
\begin{align} \label{eq_prop_tau}
\tau^{(i)}_{A^n E^n} - \tau^{(i-1)}_{A^n E^n} 
\overset{\textnormal{\Cshref{eq_def_tau}}}{=} \big( \rho_{A^i E^i} -  \rho_{A^{i-1} E^{i-1}} \otimes \proj{\theta}_{A_i E_i}  \big) \otimes \proj{\theta}_{AE}^{\otimes n -i} \, .
\end{align}
Therefore, we have
\begin{align}
\tr_{A_i E_i}[\tau^{(i)}_{A^n E^n} - \tau^{(i-1)}_{A^n E^n} ]
\overset{\textnormal{\Cshref{eq_prop_tau}}}{=} 0 \, . \label{eq_tau_are_neighboring}
\end{align}

We are now equipped with all the preliminaries to make the link to the quantum Wasserstein distance by observing that
\begin{align}
\frac{1}{n} \norm{\rho_{A^n E^n} - \proj{\theta}_{AE}^{\otimes n}}_{W_1}
\overset{\textnormal{\cite[Def.~6]{PMTL21} \& \Cshref{eq_telescoping_sum,eq_tau_are_neighboring}}}&{\leq} \frac{1}{2n} \sum_{i=1}^n \norm{\tau^{(i)}_{A^n E^n} - \tau^{(i-1)}_{A^n E^n} }_1 \\
\overset{\textnormal{\Cshref{eq_prop_tau}}}&{=} \frac{1}{2n} \sum_{i=1}^n \norm{ \rho_{A^i E^i} -  \rho_{A^{i-1} E^{i-1}} \otimes \proj{\theta}_{A_i E_i} }_1 \, . \label{eq_Wasserstein_step1_newLem}
\end{align}
To further bound this term, note that for any $i \in [n]$ we have
\begin{align}
&\hspace{-20mm}\norm{ \rho_{A^i E^i} -  \rho_{A^{i-1} E^{i-1}} \otimes \proj{\theta}_{A_i E_i} }_1 \nonumber \\
\overset{\textnormal{triangle ineq.}}&{\leq}\! \norm{ \rho_{A^i E^i} \!-\!  \eta^{(i)}_{A^{i-1} E^{i-1}} \otimes \proj{\theta}_{A_i E_i} }_1 \!+\! \norm{ (\eta^{(i)}_{A^{i-1} E^{i-1}} \!-\! \rho_{A^{i-1} E^{i-1}} ) \!\otimes\! \proj{\theta}_{A_i E_i} }_1 \\
&= \norm{ \rho_{A^i E^i} -  \eta^{(i)}_{A^{i-1} E^{i-1}} \otimes \proj{\theta}_{A_i E_i} }_1 + \norm{ \eta^{(i)}_{A^{i-1} E^{i-1}} - \rho_{A^{i-1} E^{i-1}} }_1 \\
\overset{\textnormal{contractivity~\cite[Thm.~8.16]{wolf_notes} \& \Cshref{eq_product_structure_after_projection} }}&{\leq} \norm{\rho_{A^n E^n} - \bar \rho^{(i)}_{A^n E^n}}_1 + \norm{\bar \rho^{(i)}_{A^n E^n} - \rho_{A^n E^n} }_1 \\
\overset{\textnormal{\Cshref{eq_gentle_meas_lemma}}}&{ \leq } 4 \sqrt{w_i} \, . \label{eq_Wasserstein_step2_newLem}
\end{align}
Hence, we have
\begin{align}
\frac{1}{n} \norm{\rho_{A^n E^n} - \proj{\theta}_{AE}^{\otimes n}}_{W_1}
\overset{\textnormal{\Cshref{eq_Wasserstein_step1_newLem,eq_Wasserstein_step2_newLem}}}{\leq} \frac{2}{n} \sum_{i=1}^n \sqrt{w_i} 
\overset{\textnormal{Cauchy-Schwarz}}{\leq} \frac{2}{\sqrt{n}} \sqrt{\sum_{i=1}^n w_i} 
\overset{\textnormal{\Cshref{eq_sum_weights}}}{\leq} 2 \sqrt{\frac{r}{n}} \, . \label{eq_purified_done_Wasserstein}
\end{align}

To complete the proof, we recall that the Wasserstein distance is non-increasing under taking the partial trace and therefore
\begin{align}
\frac{1}{n} \norm{\rho_{A^n} - \sigma_A^{\otimes n}}_{W_1}
\overset{\textnormal{\cite[Section~IX.A]{PMTL21}}}{\leq} \frac{1}{n} \norm{\rho_{A^n E^n} - \proj{\theta}_{AE}^{\otimes n}}_{W_1} 
\overset{\textnormal{\Cshref{eq_purified_done_Wasserstein}}}{\leq} 2 \sqrt{\frac{r}{n}} \, .
\end{align}
\end{proof}

\subsection{Entropies}
For $x \in [0,1]$ let $h(x):=-x \log x - (1-x) \log (1-x)$ denote the binary entropy function. We write $\log(\cdot)$ to denote the natural logarithm. 
Let $\cA, \cB$ be two sets of quantum states on a Hilbert space $\cH$. For an arbitrary divergence $\mathbb{D}$, we use the notation 
\begin{align}
\mathbb{D}(\cA \| \cB) := \inf_{\rho \in \cA} \inf_{\sigma \in \cB} \mathbb{D}(\rho \| \sigma) \, .
\end{align}
For $\rho, \sigma \in \St(\cH)$, the fidelity is denoted by $F(\rho,\sigma):= \norm{\sqrt{\rho} \sqrt{\sigma}}^2_1$.
The max-relative entropy~\cite{renner_phd,datta09} is defined as
\begin{align}
D_{\max}(\rho \| \sigma) :=\inf\{\lambda \in \R: \rho \leq \ee^\lambda \sigma \} \, .
\end{align} 
The min-relative entropy~\cite{dupuis12} is given by
\begin{align}
D_{\min}(\rho \| \sigma) =  - \log F(\rho, \sigma) \, .
\end{align}
For $\eps\in(0,1)$ we define the $\eps$-ball around $\rho$ by $\cB_{\eps}(\rho):=\{\rho' \in \St(\cH): P(\rho,\rho') \leq \eps \}$, where \smash{$P(\rho,\sigma):=\sqrt{1-F(\rho,\sigma)}$} is the purified distance~\cite{marco_book}. 
The smooth max-relative entropy is defined as
\begin{align}
D_{\max}^\eps(\rho  \| \sigma ) :=  \min_{\rho' \in \cB_{\eps}(\rho)} D_{\max}(\rho' \| \sigma) \, . 
\end{align} 
Analogously, the smooth min-relative entropy is given by
\begin{align}
D^\eps_{\min} (\rho \| \sigma) = \max_{\rho' \in \cB_{\eps}(\rho)} D_{\min}(\rho' \| \sigma) \, .
\end{align}
For $\alpha \in [1/2,1) \cup (1,\infty)$ the sandwiched R\'enyi relative entropy~\cite{MLDSFT13,WWY14} is given by
\begin{align}
D_{\alpha}(\rho \| \sigma) :=\frac{1}{\alpha -1 } \log \tr\big[(\sigma^{\frac{1-\alpha}{2\alpha}} \rho \, \sigma^{\frac{1-\alpha}{2\alpha}})^{\alpha}\big] \, .
\end{align}
In the limits $\alpha \to 1$, $\alpha \to \infty$, and $\alpha=\frac{1}{2}$, the sandwiched R\'enyi relative entropy converges to the relative entropy, the max-relative entropy and the min-relative entropy, respectively.

It is known that the hypothesis testing relative entropy is related to the smooth max-relative entropy. For $\rho, \sigma \in \St(\cH)$ and $\eps \in (0,1)$, $\nu \in (0,1-\eps)$ we have~\cite[Theorem~4]{anshu19}
\begin{align} \label{eq_buscemi}
D^{1-\eps}_H(\rho \| \sigma) \geq D_{\max}^{\sqrt{\eps}}(\rho \| \sigma) - \log \frac{1}{1-\eps}   \geq  D^{1-\eps-\nu}_H(\rho \| \sigma) - \log \frac{4}{\nu^2} \, .
\end{align}
This relation was originally observed by~\cite{dupuis12}.\footnote{Note that~\cite{dupuis12} uses a different convention for the hypothesis testing relative entropy where $\eps \leftrightarrow 1-\eps$.} A tighter relation was later obtained in~\cite[Theorem~12]{BLD26}, though it will not be needed in the present work.
\section{Robustness of relative entropy of resource} \label{sec_robustness_ER}
Let $\rho \in \St(\cH)$ and $\cM \subseteq \St(\cH)$ be a closed convex set of finite-dimensional states that contains a full-rank state $\omega \in \cM$  with $\omega > 0$. The \emph{relative entropy of resource}~\cite{VPRK97} is defined as
\begin{align} \label{eq_def_ER}
D(\rho\| \cM) :=  \min_{\sigma \in \cM} D(\rho \| \sigma) \, .
\end{align}
An important example is the case where $\cH = \cH_{AB}$ and $\cM= \mathrm{SEP}(A:B)$, in which case the relative entropy of resource reduces to the relative entropy of entanglement.

We have seen in~\cref{eq_W1_to_TraceDist} that the Wasserstein distance is a weaker metric than the trace distance. This can be relevant because almost-iid states and iid states are close in Wasserstein distance but not close in trace distance (as explained in~\cref{subsec_almost_iid}). 
Many entropic quantities of interest (e.g.~the entropy or the relative entropy of resource) have been shown to be continuous with respect to the trace distance~\cite{win16}. Interestingly, the entropy is even continuous with respect to the weaker Wasserstein distance~\cite[Theorem~9.1]{PT_23}. This raises the question whether the relative entropy of resource is also continuous with respect to the Wasserstein distance. We next show that this is indeed the case.    

For some density matrix $\omega \in \St(\cH)$, let $\cR^{(\omega)}(X) := \tr[X] \omega$ denote the replacer channel. For $i \in [n]$ we use the notation \smash{$\cR^{(\omega)}_{i} := \cI_{1} \otimes \ldots \cI_{i-1} \otimes \cR \otimes \cI_{i+1} \otimes \ldots \otimes \cI_n$}. For $n \in \N$, we consider a set $\cM_n$ that satisfies the following two axioms:
\begin{enumerate}[(a)]
\item $\cM_n$ is a convex closed subset of $\St(\cH^{\otimes n})$; \label{axiom_Giulia_1}
\item $\exists \, \omega \in \St(\cH)$ with full support (i.e.~$\lambda_{\min}(\omega) >0$) such that $\sigma_n \in \cM_n$ implies $\cR^{(\omega)}_i(\sigma_n) \in \cM_n$.\label{axiom_Giulia_3}
\end{enumerate}
Note that the set of separable states, i.e.~$\cM_n=\mathrm{SEP}(A^n:B^n)$ satisfies Axioms~\eqref{axiom_Giulia_1}-\eqref{axiom_Giulia_3}.\footnote{By choosing $\omega = \frac{\id_{AB}}{d_{AB}}$.}

\begin{proposition}[Continuity with respect to Wasserstein distance] \label{prop_continuity_ER_Wasserstein}
Let $n \in \N$, $\rho_n, \rho'_n \in \St(\cH^{\otimes n})$ with $d = \dim(\cH)$, and $\cM_n$ be a set that satisfies Axioms~\eqref{axiom_Giulia_1}-\eqref{axiom_Giulia_3}.
If $\frac{1}{n} \norm{\rho_n - \rho'_n}_{W_1}\leq \eps_n \leq \frac{1}{2}$, then 
\begin{align}
\Big| \frac{1}{n} D(\rho_n \| \cM_n)  - \frac{1}{n} D(\rho'_n \| \cM_n) \Big| 
\leq 3 h(\eps_n) + 6 \eps_n  \log \frac{d}{\lambda_{\min}(\omega)} \, .
\end{align}
\end{proposition}
\begin{proof}
For $t\geq 0$ define the channel
\begin{align} \label{eq_cE_t_channel}
\cE^{(t)}(X) := \ee^{-t} X + \big(1-\ee^{-t}\big) \tr[X] \omega= \ee^{-t} \cI(X) + \big(1-\ee^{-t}\big) \cR^{(\omega)}(X)\, ,
\end{align}
and consider
\begin{align} \label{eq_channel_E}
\bar \cE^{(t)} := \bigotimes_{i=1}^n \cE^{(t)} \, .
\end{align}
Let $\hat \sigma_n \in \arg\min_{\sigma_{n} \in \cM_n} D(\rho_{n} \| \sigma_{n})$ be such that
\begin{align} \label{eq_ER1}
D(\rho_n \| \cM_n)  =- H(\rho_n) - \tr[\rho_n \log \hat \sigma_n] \ .
\end{align}
Note that clearly we have $\supp(\rho_{n}) \subseteq \supp(\hat \sigma_{n})$ as otherwise the relative entropy is unbounded.
For any $i \in [n]$ and any $\tau_n \in \cM_n$ we have
\begin{align}
(\cI_1 \otimes \ldots \otimes \cI_{i-1} \otimes \cE^{(t)} \otimes \cI_{i+1} \otimes \ldots \otimes \cI_n)(\tau_n)  
= \ee^{-t} \tau_n + (1-\ee^{-t}) \cR_i^{(\omega)}(\tau_n) \in \cM_n \, ,
\end{align}
where we used Axioms~\eqref{axiom_Giulia_1}-\eqref{axiom_Giulia_3}. Hence, we have
\begin{align}
\hat \sigma^{(t)}_n:= \bar \cE^{(t)}(\hat \sigma_n) = (\cE^{(t)} \otimes \cI_{2} \otimes \ldots \otimes \cI_{n}) \circ \ldots \circ (\cI_{1} \otimes \ldots \otimes \cI_{n-1} \otimes \cE^{(t)}) (\hat \sigma_n) \in \cM_n \, . \label{eq_sigma_in_set}
\end{align}
Therefore, we have
\begin{align}  \label{eq_ER2}
D(\rho'_n \| \cM_n)  
\overset{\textnormal{\Cshref{eq_sigma_in_set}}}&{\leq} D(\rho'_n \| \hat \sigma^{(t)}_n  )
= - H(\rho'_n) - \tr[\rho'_n \log \hat \sigma^{(t)}_n]
\end{align}
and 
\begin{align} \label{eq_3_terms_to_bound}
&\hspace{-20mm}D(\rho'_n \| \cM_n)   - D(\rho_n \| \cM_n)   \nonumber \\
\overset{\textnormal{\Cshref{eq_ER1}\&\eqref{eq_ER2}}}&{\leq} \Big( H(\rho_n) - H(\rho'_n) \Big) +\tr[(\rho_n - \rho'_n) \log \hat \sigma^{(t)}_n] + \tr[\rho_n (\log \hat\sigma_n - \log \hat \sigma^{(t)}_n)] \, .
\end{align}
We next bound all three terms on the right-hand side of~\cref{eq_3_terms_to_bound} separately.

We start by recalling that the von Neumann entropy is continuous with respect to the Wasserstein distance~\cite[Theorem~9.1]{PT_23} and hence
\begin{align}
 H(\rho_n) -  H(\rho'_n) 
 \overset{\textnormal{\cite[Theorem~9.1]{PT_23}}}{\leq}  n h(\eps_n) + n \eps_n \log(d^2 -1)
 \leq  n h(\eps_n) + 2 n \eps_n \log d \, . \label{eq_STEP1_done}
\end{align}

We next bound the second term in~\cref{eq_3_terms_to_bound}. To do so, we first note that
\begin{align}
\tr[(\rho_n - \rho'_n) \log \hat \sigma^{(t)}_n]
\overset{\textnormal{\cite[Prop.~9]{PMTL21}}}&{\leq} \norm{\rho_n - \rho'_n}_{W_1} \norm{\log \hat \sigma^{(t)}_n}_L
\leq n \eps_n \norm{\log \hat \sigma^{(t)}_n}_L \, , \label{eq_second_term_1}
\end{align}
where $\| \log \hat \sigma^{(t)}_n \|_L $ is the quantum Lipschitz constant of the operator $\log \hat \sigma^{(t)}_n$ defined in~\cite[Definition~8]{PMTL21}. According to~\cite[Proposition~8]{PMTL21} this constant can be expressed as 
\begin{align}
\norm{\log \hat \sigma^{(t)}_n}_L  
\overset{\textnormal{\cite[Prop.~8]{PMTL21}}}{=} 2 \max_{i \in [n]} \min_{X_{n-1} \in \mathrm{Herm}(\cH^{\otimes n-1})} \norm{\log \hat \sigma^{(t)}_n - X_{n-1} \otimes \id_i}_{\infty} \, , \label{eq_Lipschitz}
\end{align}
where $X_{n-1}$ is an arbitrary Hermitian operator that does not act on the $i$-th subsystem. For any $i \in [n]$ we can write
\begin{align}
\hat \sigma^{(t)}_n
&= \bar \cE^{(t)} (\hat \sigma_n) \\
\overset{\textnormal{\Cshref{eq_channel_E}}}&{=} \cE_i^{(t)} \circ \underbrace{(\cE^{(t)} \otimes \ldots \otimes \cE^{(t)} \otimes \cI_{i} \otimes \cE^{(t)} \otimes \ldots \otimes \cE^{(t)}) (\hat \sigma)}_{=:\tilde \sigma^{(t)}_n} \\
&=\ee^{-t} \tilde \sigma^{(t)}_n + (1-\ee^{-t}) \cR^{(\omega)}_{i}(\tilde \sigma^{(t)}_n) \, .\label{eq_step1_cont_ER}
\end{align}
 Noting that $\ee^{-t} \tilde \sigma^{(t)}_n  \geq 0$ yields
\begin{align}  \label{eq_ds1_1}
\hat \sigma^{(t)}_n   \overset{\textnormal{\Cshref{eq_step1_cont_ER}}}{\geq}(1-\ee^{-t}) \cR^{(\omega)}_{i}(\tilde \sigma^{(t)}_n) \, .
\end{align}
We next observe that for any $\tau_{AB} \in \St(\cH_{AB})$ we have $\tau_{AB} \leq d_A (\id_A \otimes \tau_B)$.\footnote{To see this inequality, first note that due to convexity, it suffices to prove the inequality for pure states $\tau_{AB}$, which can always be written as $\tau_{A B} = d \sqrt{\tau_B} \proj{\Phi} \sqrt{\tau_B}$, where $\ket{\Phi}$ is a maximally entangled state between $A$ and $B$ and $d=\min\{d_A,d_B\}$. The inequality then follows by sandwiching the operator inequality $\proj{\Phi} \leq \id_A \otimes \id_B$ with $\sqrt{\tau_B}$.} Therefore, 
\begin{align}
\tilde \sigma^{(t)}_n
\leq d (\id_i \otimes \tr_{i}[\tilde \sigma^{(t)}_n]) 
\leq \frac{d}{\lambda_{\min}(\omega)} \cR^{(\omega)}_i(\tilde \sigma^{(t)}_n) \,  \label{eq_partial_trace_ineq}
\end{align}
which gives
\begin{align}
\hat \sigma_n^{(t)} 
\overset{\textnormal{\Cshref{eq_step1_cont_ER}}}&{=} \ee^{-t} \tilde \sigma^{(t)}_n + (1-\ee^{-t}) \cR^{(\omega)}_i(\tilde \sigma^{(t)}_n)  \\
\overset{\textnormal{\Cshref{eq_partial_trace_ineq}}}&{\leq} \Big(1+ \ee^{-t}\big(\frac{d}{\lambda_{\min}(\omega)}  -1\big) \Big)  \cR^{(\omega)}_i(\tilde \sigma^{(t)}_n) \\
&\leq \frac{d}{\lambda_{\min}(\omega)}  \cR^{(\omega)}_i(\tilde \sigma^{(t)}_n)  \, .   \label{eq_ds1_2}  
\end{align}
Using that $x \to \log x$ is operator monotone~\cite[Table~2.2]{Sutter_book} yields 
\begin{align}
\log(1-\ee^{-t}) + \log \cR^{(\omega)}_i(\tilde \sigma^{(t)}_n)
\overset{\textnormal{\Cshref{eq_ds1_1}}}{\leq} \log \hat \sigma^{(t)}_n 
\overset{\textnormal{\Cshref{eq_ds1_2}}}{\leq} \log \cR^{(\omega)}_i(\tilde \sigma^{(t)}_n) + \log \frac{d}{\lambda_{\min}(\omega)} 
\end{align}
which can be rewritten as
\begin{align}
\log(1-\ee^{-t})  \leq \log \hat \sigma^{(t)}_n - \log \cR^{(\omega)}_i(\tilde \sigma^{(t)}_n) \leq  \log \frac{d}{\lambda_{\min}(\omega)}  \, . \label{eq_ds1_3}  
\end{align}
Thus, we find
\begin{align}
\norm{\log \hat \sigma^{(t)}_n - \log \cR^{(\omega)}_i(\tilde \sigma^{(t)}_n)}_{\infty}
 \overset{\textnormal{\Cshref{eq_ds1_3}}}&{\leq} \max \Big\{ \big|\log(1-\ee^{-t})  \big|, \log\frac{d}{\lambda_{\min}(\omega)}  \Big\} \\
 &\leq \max\Big\{  \frac{t}{2}-\log t, \log \frac{d}{\lambda_{\min}(\omega)}  \Big\} \, . \label{eq_second_term_2}
\end{align}
Define the Hermitian matrix on the $n-1$ subsystems as
\begin{align} \label{eq_guess_X}
\bar X_{n-1} := \log \tr_{i}[\tilde \sigma^{(t)}_n] 
\end{align}
and note that
\begin{align} \label{eq_operator_ineq_new_resource}
\bar X_{n-1} \otimes \id_i + \log(\lambda_{\min}(\omega)) \id 
\leq \log \cR^{(\omega)}_i(\tilde \sigma^{(t)}_n)
\leq \bar X_{n-1} \otimes \id_i \, ,
\end{align}
which can be rewritten as
\begin{align} \label{eq_operator_ineq_new_resource_2}
 \log(\lambda_{\min}(\omega)) \id 
\leq \log \cR^{(\omega)}_i(\tilde \sigma^{(t)}_n) - \bar X_{n-1} \otimes \id_i
\leq 0\, .
\end{align}
To see this recall that $\cR^{(\omega)}_i(\tilde \sigma^{(t)}_n) = \tr_{i}[\tilde\sigma^{(t)}_n] \otimes \omega = \bar X_{n-1} \otimes \omega$ and hence~\cref{eq_operator_ineq_new_resource} follows from $\lambda_{\min}(\omega) \id \leq \omega \leq \id $ together with the fact that $\log(Y \otimes Z) = \log(Y) \otimes \id + \id \otimes \log(Z)$ and the operator monotonicity of the logarithm~\cite[Table~2.2]{Sutter_book}.
Using that $\bar X_{n-1}$ defined in~\cref{eq_guess_X} is Hermitian, we find
\begin{align}
\norm{\log \hat \sigma^{(t)}_n}_L 
\overset{\textnormal{\Cshref{eq_Lipschitz}}}&{=} 2 \max_{i \in [n]} \min_{X_{n-1} \in \mathrm{Herm}(\cH^{\otimes n-1})} \norm{\log \hat \sigma^{(t)}_n - X_{n-1} \otimes \id_i}_{\infty}  \\
&\leq  2 \max_{i \in [n]} \norm{\log \hat \sigma^{(t)}_n - \bar X_{n-1} \otimes \id_i}_{\infty} \\
\overset{\textnormal{triangle ineq.}}&{\leq} 2 \max_{i \in [n]} \norm{\log \hat \sigma^{(t)}_n - \log \cR^{(\omega)}_i(\tilde \sigma^{(t)}_n)}_{\infty} + 2 \max_{i \in [n]} \norm{ \log \cR^{(\omega)}_i(\tilde \sigma^{(t)}_n) - \bar X_{n-1} \otimes \id_i}_{\infty}\\
\overset{\textnormal{\Cshref{eq_second_term_2,eq_operator_ineq_new_resource_2}}}&{\leq}  \max\Big\{  t - 2 \log t, 2 \log \frac{d}{\lambda_{\min}(\omega)} \Big\} -  2\log \lambda_{\min}(\omega) \\
&\leq  t - 2 \log t + 2 \log \frac{d}{\lambda_{\min}(\omega)}  -  2\log \lambda_{\min}(\omega) \\
&= t - 2 \log t + 2 \log d - 4 \log \lambda_{\min}(\omega) \, . \label{eq_second_term_3}
\end{align}
We thus conclude the following bound
\begin{align} \label{eq_STEP2_done}
\tr[(\rho_n - \rho'_n) \log \hat \sigma_n^{(t)}]  
\overset{\textnormal{\Cshref{eq_second_term_1,eq_second_term_3}}}{\leq} n \eps_n \big(t - 2 \log t + 2 \log d - 4 \log \lambda_{\min}(\omega) \big)  \, .
\end{align}

It remains to bound the third and final term in~\cref{eq_3_terms_to_bound}. To do so note that by binomial expansion we can write
\begin{align}
\bar \cE^{(t)} 
\overset{\textnormal{\Cshref{eq_channel_E}}}{=} \bigotimes_{i=1}^n \cE^{(t)} 
\overset{\textnormal{\Cshref{eq_cE_t_channel}}}{=} \bigotimes_{i=1}^n \left( \ee^{-t} \cI + \big(1-\ee^{-t}\big) \cR^{(\omega)} \right)
=:\ee^{-n t} \cI_1 \otimes \ldots \otimes \cI_n + (1-\ee^{-n t} ) \cF^{(t)} \, , \label{eq_channel_E_dec}
\end{align}
for some trace-preserving and completely positive map $\cF^{(t)}$ acting on all $n$ subsystems. Hence,
\begin{align}
\hat \sigma_n^{(t)} 
= \bar \cE^{(t)} (\hat \sigma_n)
\overset{\textnormal{\Cshref{eq_channel_E_dec}}}{=} \ee^{-nt} \hat \sigma_n +  (1-\ee^{-n t} ) \cF^{(t)}(\hat \sigma_n)
\geq \ee^{-nt} \hat \sigma_n \, . \label{eq_almost_done_step3}
\end{align}
Using that $x \mapsto \log x$ is operator monotone~\cite[Table~2.2]{Sutter_book} yields
\begin{align} \label{eq_log_support}
\log \hat \sigma_n - \log \hat \sigma_n^{(t)} \leq \id nt   \, ,
\end{align}
hence taking an expectation value with respect to $\rho_n$ yields
\begin{align} \label{eq_STEP3_done}
\tr[\rho_n (\log \hat\sigma_n - \log \hat \sigma_n^{(t)})]  \leq n t \, .
\end{align}
Recall here that, as explained above, we have $\supp(\rho_{n}) \subseteq \supp(\hat \sigma_{n})$ and $\hat \sigma^{(t)}_n$ has full support. In the case where $\hat \sigma_{n}$ does not have full support, we view~\eqref{eq_log_support} restricted to the support of $\hat \sigma_{n}$.

Choosing $t=\eps_n$ and  putting everything together yields
\begin{align}
&\hspace{-23mm}\frac{1}{n} D(\rho'_n \| \cM_n)   -  \frac{1}{n} D(\rho_n \| \cM_n) \nonumber    \\
\overset{\textnormal{\Cshref{eq_3_terms_to_bound,eq_STEP1_done,eq_STEP2_done,eq_STEP3_done}}}&{\leq}  h(\eps_n) + 2 \eps_n \log d + \eps_n \big(\eps_n - 2 \log \eps_n + 2 \log d - 4 \log \lambda_{\min}(\omega) \big) + \eps_n \\
&\leq 3 h(\eps_n) + 6 \eps_n  \log \frac{d}{\lambda_{\min}(\omega)} \, ,
\end{align}
where in the final step, we used $-\eps_n \log \eps_n \leq h(\eps_n)$ and $\eps_n \leq \frac{1}{2}$. 
We can swap the roles of $\rho$ and $\rho'$ and repeat the entire argument above, which yields
\begin{align}
\frac{1}{n} D(\rho_n \| \cM_n)   -  \frac{1}{n} D(\rho'_n \| \cM_n) 
\leq 3 h(\eps_n) + 6 \eps_n  \log \frac{d}{\lambda_{\min}(\omega)}  \, ,
\end{align}
and thus completes the proof.
\end{proof}

\begin{example}[Relative entropy of entanglement]
Consider the case $\cH = \cH_{AB}$ with $d_{AB} = \dim \cH_{AB}$, $\cM_n= \mathrm{SEP}(A^n:B^n)$ and $\omega = \frac{\id_{AB}}{d_{AB}}$. Then for $\rho_{A^n B^n}, \rho'_{A^n B^n} \in \St(\cH_{AB}^{\otimes n})$ such that $\frac{1}{n} \norm{\rho_n - \rho'_n}_{W_1}\leq \eps_n \leq \frac{1}{2}$ we have
\begin{align}
\Big| \frac{1}{n} D(\rho \| \mathrm{SEP})  - \frac{1}{n} D(\rho' \| \mathrm{SEP}) \Big| 
\overset{\textnormal{\Cshref{prop_continuity_ER_Wasserstein}}}{\leq} 3 h(\eps_n) + 12 \eps_n  \log d_{AB} \, .
\end{align}
\end{example}

\begin{corollary} \label{cor_asymptotic_continuity_ER}
Let $n,r \in \N$ be such that $16r \leq n$, $\sigma \in \St(\cH)$, $\rho_n \in \St^n(\cH, \sigma^{\otimes n -r})$, $d = \dim(\cH)$, and let $\cM_n$ be a set that satisfies Axioms~\eqref{axiom_Giulia_1}-\eqref{axiom_Giulia_3}. Then 
\begin{align}
\Big| \frac{1}{n} D(\rho_n \| \cM_n)  - \frac{1}{n} D(\sigma^{\otimes n} \| \cM_n) \Big|  \leq 3 h\Big(2 \sqrt{\frac{r}{n}}\Big) + 12 \sqrt{\frac{r}{n}}  \log \frac{d}{\lambda_{\min}(\omega)}  \, .
\end{align}
 \end{corollary}
\begin{proof}
Using the fact that almost-iid states are close to iid states in terms of Wasserstein distance, we find
\begin{align}
\frac{1}{n} \norm{\rho_{n} - \sigma^{\otimes n}}_{W_1}
\overset{\textnormal{\Cshref{lem_amost_iid_close_to_iid_W1}}}{\leq}  2 \sqrt{\frac{r}{n}} \, .
\end{align}
Then
\begin{align}
\Big| \frac{1}{n} D(\rho_n \| \cM_n)  - \frac{1}{n} D(\sigma^{\otimes n} \| \cM_n)  \Big| 
\overset{\textnormal{\Cshref{prop_continuity_ER_Wasserstein}}}{\leq} 3 h\Big( 2 \sqrt{\frac{r}{n}} \Big) + 12 \sqrt{\frac{r}{n}}  \log \frac{d}{\lambda_{\min}(\omega)} \, .
\end{align}
\end{proof}

\begin{remark}
For $\omega>0$, the set of generalized almost-iid states $\cM_n = \bar \St^n(\cH,\omega^{\otimes n-r})$ satisfies  Axioms~\eqref{axiom_Giulia_1}-\eqref{axiom_Giulia_3}. This is proven in~\cref{lem_generalized_almost_iid_satisfies_axioms}.
The set of almost-iid states, i.e.~$\cM_n = \St^n(\cH,\sigma^{\otimes n-r})$ does not satisfy Axiom~\eqref{axiom_Giulia_3} because the map \smash{$\cR_i^{(\omega)}$} can break the permutation invariance. However, the statements of~\cref{prop_continuity_ER_Wasserstein,cor_asymptotic_continuity_ER} still hold. The reason is that the map $\bar \cE^{(t)}$ used in the proof of~\cref{prop_continuity_ER_Wasserstein} preserves the permutation invariance.
\end{remark}

\section{Robust (generalized) quantum Stein's lemma} \label{sec_robust_GSL}
In this section, we state and prove a \emph{robust generalized quantum Stein's lemma} (see~\cref{thm_AGSL}) and a \emph{robust quantum Stein's lemma} (see~\cref{thm_ASL}).

\subsection{Robust generalized quantum Stein's lemma}
Let $\cS_n$ be the set of permutations on $\{ 1,... ,n\}$ and consider the trace-preserving completely positive map
\begin{align}
\mathrm{SYM}_n: X \mapsto \frac{1}{n!}\sum_{\pi \in \cS_n} \pi X \pi^\dagger \, ,
\end{align}
 which symmetrizes the input. 
Let $\{\cM_n\}_{n \in \N}$ be a family of sets that satisfy the following five axioms~\cite{brandao_Stein_10}:
\begin{enumerate}[(i)]
\item each $\cM_n$ is a convex and closed subset of $\St(\cH^{\otimes n})$; \label{axiom_convexity}
\item $\cM_1$ contains a full-rank state $\omega$, i.e. $\cM_1 \ni \omega > 0$; \label{axiom_full_rank}
\item $\rho_n \in \cM_n$ implies that $\pi \rho_n \pi^\dagger \in \cM_n$ for any permutation $\pi \in \cS_n$; \label{axiom_PI}
\item $\rho_{n+1} \in \cM_{n+1}$ implies $\rho_n \in  \cM_n$; \label{axiom_partial_trace}
\item $\rho_n \in \cM_n$ and $\sigma_m \in \cM_m$ imply $\rho_n \otimes \sigma_m \in \cM_{n+m}$. \label{axiom_tensor_product}
\end{enumerate}
\begin{example} \label{ex_axioms}
The following two families of sets satisfy Axioms~\eqref{axiom_convexity}-\eqref{axiom_tensor_product}:
\begin{enumerate}
\item Set of separable states, i.e.~$\cM_n = \mathrm{SEP}(A^n:B^n)$;
\item Set of a single iid state, i.e.~$\cM_n = \{\sigma^{\otimes n} \}$ for $\sigma>0$.
\end{enumerate}
\end{example}
Any family of sets $\{\cM_n\}_{n \in \N}$ that satisfies Axioms~\eqref{axiom_convexity}-\eqref{axiom_tensor_product} also satisfies Axioms~\eqref{axiom_Giulia_1}-\eqref{axiom_Giulia_3}. To see this, note that the first axioms coincide. In addition, we have $\omega \in \cM_1$ (by Axiom~\eqref{axiom_full_rank}) and and $\tr_i[\sigma_n] \in \cM_{n-1}$ for $\sigma_n \in \cM_n$ (by Axiom~\eqref{axiom_partial_trace}). Therefore, we conclude
\begin{align}
\cR_i^{(\omega)}(\sigma_ n) = \omega \otimes \tr_i[\sigma_n] \quad  \overset{\textnormal{Axiom~\eqref{axiom_tensor_product}}}{\implies}  \quad \cR_i^{(\omega)}(\sigma_ n) \in \cM_n \, .
\end{align}

For $\rho \in \St(\cH)$, recall the definition of the relative entropy of resource~\cite{VPRK97} 
\begin{align}
D(\rho\| \cM) :=  \min_{\sigma \in \cM} D(\rho \| \sigma) \, .
\end{align}
Its regularized version is given by
\begin{align} \label{eq_def_regul_2}
D^\reg(\rho \| \cM) 
:=\lim_{n \to \infty} \frac{1}{n} D(\rho^{\otimes n} \| \cM_n)
=\inf_{n \in \N} \frac{1}{n} D(\rho^{\otimes n} \| \cM_n)\, ,
\end{align}
which is well-defined as ensured by Fekete's lemma (see~\cref{lem_fekete}).
For $r=o(n)$, following the notation introduced above, we define the \emph{almost-iid regularized relative entropy of resource} as
\begin{align} \label{eq_def_A_RE}
\lim_{n \to \infty} \frac{1}{n} D\big( \St^n(\cH, \rho^{\otimes n-r}) \| \cM_n \big) = \lim_{n \to \infty} \frac{1}{n} \min_{\rho_n \in \St^n(\cH, \rho^{\otimes n-r})} D(\rho_n \| \cM_n) \, .
\end{align}
We show in~\cref{lem_regul_almost_iid_irrelevant_Stein} below that the quantity in~\cref{eq_def_A_RE} is well-defined in the sense that the limit exists. Furthermore, the quantity is independent of $r$ as long as $r=o(n)$.

\begin{lemma} \label{lem_regul_almost_iid_irrelevant_Stein}
Let $\rho \in \St(\cH)$ and $\{\cM_n\}_{n \in \N}$  be a family of sets satisfying properties~\eqref{axiom_convexity}-\eqref{axiom_tensor_product}, and $r=o(n)$. Then
\begin{align}
\lim_{n \to \infty} \frac{1}{n} \D\big( \St^n(\cH, \rho^{\otimes n-r}) \| \cM_n \big) = D^\reg(\rho \| \cM)  \, .
\end{align}
\end{lemma}
\begin{proof}
Note that
\begin{align}
\lim_{n \to \infty} \frac{1}{n} \D\big( \St^n(\cH, \rho^{\otimes n-r}) \| \cM_n \big)
\overset{\textnormal{\Cshref{eq_def_A_RE}}}&{=}\lim_{n \to \infty} \frac{1}{n} \min_{\rho_n \in \St^n(\cH, \rho^{\otimes n-r})}  D(\rho_n \| \cM_n) \\
\overset{\textnormal{\Cshref{cor_asymptotic_continuity_ER}}}&{=}\lim_{n \to \infty} \frac{1}{n} D(\rho^{\otimes n} \| \cM_n) \\
\overset{\textnormal{\Cshref{eq_def_regul_2}}}&{=}D^\reg(\rho \| \cM)  \, .
\end{align}
\end{proof}

We next prove in~\cref{thm_AGSL} a robust and hence stronger version of the generalized quantum Stein's lemma, where the null hypothesis is only almost-iid with defects of size $r=o(n)$ instead of iid. Choosing $r=0$ yields the traditional generalized quantum Stein's lemma $\lim_{n \to \infty}\frac{1}{n}   D_H^{\eps}(\rho^{\otimes n} \| \cM_n)
 = D^\reg(\rho \| \cM)$ from~\cite{brandao_Stein_10,haya_stein_25,ludo25}.
Our proof follows the original proof construction given in~\cite{brandao_Stein_10}.  
As discussed in~\cite{Berta2023gapinproofof,stein_nature_24}, the proof of~\cite[Lemma~III.9]{brandao_Stein_10} contains a logical error that invalidates~\cite[Lemma~III.7]{brandao_Stein_10}. However, \Cref{cor_asymptotic_continuity_ER} is a stronger version of~\cite[Lemma~III.7]{brandao_Stein_10} and thus fixes the logical gap in the original proof of the generalized quantum Stein's lemma.
In addition, our proof generalizes some steps which are necessary, as our null hypothesis is not iid. Furthermore, we polish and simplify some steps of~\cite{brandao_Stein_10} using modern tools from one-shot information theory that were not developed at the time of~\cite{brandao_Stein_10}.

\begin{theorem}[Robust generalized quantum Stein's lemma] \label{thm_AGSL}
Let $\eps \in (0,1)$, $\rho \in \St(\cH)$, $\{\cM_n\}_{n \in \N}$ be a family of sets that satisfy the Axioms~\eqref{axiom_convexity}-\eqref{axiom_tensor_product}, and $r=o(n)$. Then
\begin{align} \label{eq_RGQSL}
 \lim_{n \to \infty}\frac{1}{n}   D_H^{\eps}\big( \St^n(\cH, \rho^{\otimes n -r}) \| \cM_n \big)
 = D^\reg(\rho \| \cM) \, .
\end{align}
\end{theorem}
We prove the two directions of~\cref{eq_RGQSL} separately. Recalling the operational interpretation of the hypothesis testing relative entropy defined in~\cref{eq_def_DH}, the two directions correspond to the following two statements:
\begin{enumerate}[(i)]
\item there exists a sequence $\{M_n,\id-M_n \}_{n \in \N}$ of POVMs that has an arbitrarily small type I error and achieves an error exponent $D^\reg(\rho \| \cM)$; \label{it_achievability}
\item the error exponent $D^\reg(\rho \| \cM)$ is optimal.  \label{it_strong_converse}
\end{enumerate}
Part~\eqref{it_achievability} is called ``achievability" whereas part~\eqref{it_strong_converse} is usually referred to as ``converse".\footnote{Here we even prove a ``strong converse" result as~\cref{eq_RGQSL} is true for any fixed $\eps \in (0,1)$.}
\begin{proof}
\textbf{Strong converse:}
We first note that $\rho^{\otimes n} \in \St^n(\cH, \rho^{\otimes n -r})$ and hence
\begin{align} \label{eq_strong_converse_step0}
D_H^{\eps}\big( \St^n(\cH, \rho^{\otimes n -r}) \| \cM_n \big) 
\leq D_H^{\eps}\big(  \rho^{\otimes n} \| \cM_n \big) \, .
\end{align}
As a result, the strong converse for the robust generalized quantum Stein's lemma follows from the strong converse for the generalized quantum Stein's lemma~\cite{brandao_Stein_10}. We repeat the proof here for completeness and readability, but note that it follows from standard techniques in one-shot information theory~\cite{marco_book}. Furthermore, the strong converse proof given in~\cite{haya_stein_25} works similarly. 
For any $\alpha > 1$ and choosing $\nu = \frac{1-\eps}{2}$ in~\cref{eq_buscemi}, we can write
\begin{align}
 \lim_{n \to \infty}\frac{1}{n}   D_H^{\eps}( \rho^{\otimes n} \| \cM_n )
\overset{\textnormal{\Cshref{eq_buscemi}}}&{\leq} \lim_{n \to \infty}  \frac{1}{n}  D_{\max}^{\sqrt{\frac{1-\eps}{2}}}( \rho^{\otimes n} \| \cM_n)   \\
\overset{\textnormal{\cite[Thm.~3]{anshu19}}}&{\leq}\!\!\! \lim_{n \to \infty}  \frac{1}{n}   D_{\alpha}( \rho^{\otimes n} \| \cM_n)   \\
\overset{\textnormal{\Cshref{lem_fekete}}}&{=} \inf_{n \in \N}  \frac{1}{n}  D_{\alpha}( \rho^{\otimes n} \| \cM_n)  \, .\label{eq_AGSL_1}
\end{align} 
Since the above steps hold for all $\alpha > 1$, we obtain
\begin{align}
 \lim_{n \to \infty} \frac{1}{n} D_H^{\eps}\big( \St^n(\cH, \rho^{\otimes n -r}) \| \cM_n \big) 
\overset{\textnormal{\Cshref{eq_AGSL_1,eq_strong_converse_step0}}}&{\leq} \inf_{\alpha>1}  \inf_{n \in \N}  \frac{1}{n}     D_{\alpha}( \rho^{\otimes n} \| \cM_n)   \\
&=   \inf_{n \in \N}   \min_{\sigma_n \in \cM_n}   \inf_{\alpha>1}  \frac{1}{n}  D_{\alpha}(\rho^{\otimes n} \| \sigma_n )  \\
\overset{\textnormal{\cite[Cor.~4.13]{marco_book}}}&{=}   \inf_{n \in \N}  \frac{1}{n}  \min_{\sigma_n \in \cM_n}   \lim_{\alpha \searrow 1}  D_{\alpha}(\rho^{\otimes n} \| \sigma_n ) \\
\overset{\textnormal{\cite[Prop.~4.15]{marco_book}}}&{=}   \inf_{n \in \N}  \frac{1}{n} \min_{\sigma_n \in \cM_n}   D(\rho^{\otimes n} \| \sigma_n ) \\
\overset{\textnormal{\Cshref{eq_def_regul_2}}}&{=} D^\reg(\rho \| \cM) \, , 
\end{align}
which completes the proof of the strong converse part.

\textbf{Achievability:} 
Our argument follows the original proof strategy developed in~\cite{brandao_Stein_10}. As in~\cite{ludo25}, we phrase it in terms of modern tools from one-shot information theory, such as the smooth max-relative entropy, which were not yet fully developed at the time of~\cite{brandao_Stein_10}. In addition, since we prove a robust generalized quantum Stein’s lemma with an almost-iid null hypothesis, several steps from~\cite{brandao_Stein_10} require nontrivial generalization, as the original argument only treats the iid case.

First, observe that
\begin{align}
\frac{1}{n}   D_H^{\eps}\big( \St^n(\cH, \rho^{\otimes n -r}) \| \cM_n \big)
\overset{\textnormal{\Cshref{eq_buscemi}}}&{\geq}  \frac{1}{n} D_{\max}^{\sqrt{1-\eps}}\big( \St^n(\cH, \rho^{\otimes n -r}) \| \cM_n \big) - \frac{1}{n} \log \frac{1}{\eps} \\
&=  \min_{\rho_n \in \St^n(\cH, \rho^{\otimes n -r}) } \min_{\sigma_n \in  \cM_n}  \frac{1}{n} D_{\max}^{\sqrt{1-\eps}}\big( \rho_n \|\sigma_n \big) - o(1)  \, . \label{eq_relate_Dmax_achievability}
\end{align} 
Let $\hat \rho_n$ and $\hat \sigma_n$ denote the optimizers of the right-hand side of~\cref{eq_relate_Dmax_achievability}.
 The DPI for the smooth max-relative entropy (see~\cref{lem_DPI_Dmax_eps}) implies that
 \begin{align}
D_{\max}^{\sqrt{1-\eps}}(\hat \rho_n \| \hat \sigma_n) 
\overset{\textnormal{\Cshref{lem_DPI_Dmax_eps}}}&{\geq} D_{\max}^{\sqrt{1-\eps}}\big(\mathrm{SYM}_n(\hat \rho_n) \| \mathrm{SYM}_n(\hat \sigma_n) \big)
=D_{\max}^{\sqrt{1-\eps}}\big(\hat \rho_n \| \mathrm{SYM}_n(\hat \sigma_n) \big) \, ,
 \end{align}
 where the final step uses that almost-iid states are permutation invariant (see~\cref{def_almost_product_state_mixed}).
Axioms~\eqref{axiom_convexity} and~\eqref{axiom_PI} imply that for $\hat \sigma_n \in  \cM_n$ we have $\mathrm{SYM}_n(\hat \sigma_n) \in \cM_n$. This shows that the optimizer over $\cM_n$ on the right-hand side of~\cref{eq_relate_Dmax_achievability} can be assumed to be permutation invariant. In the following, we denote it by $\bar \sigma_n$. Let $\bar \rho_n \in \cB_{\sqrt{1-\eps}}(\hat \rho_n)$ be such that 
\begin{align}
\min_{\rho_n \in \St^n(\cH, \rho^{\otimes n -r}) } \min_{\sigma_n \in  \cM_n}  \frac{1}{n} D_{\max}^{\sqrt{1-\eps}}\big( \rho_n \|\sigma_n \big)
= \frac{1}{n} D_{\max}^{\sqrt{1-\eps}}\big( \hat \rho_n \| \bar \sigma_n \big)
= \frac{1}{n} D_{\max}\big( \bar \rho_n \| \bar \sigma_n \big)
=:y \, . \label{eq_DEF_y}
\end{align}
By definition of the max-relative entropy, we have
\begin{align}\label{eq_not_yet_PI}
\bar \rho_n \leq \ee^{yn} \bar \sigma_n \qquad \textnormal{and} \qquad F(\bar \rho_n,\rho_n) \geq \eps \, .
\end{align}
Next, we show that $\tilde \rho_n := \mathrm{SYM}_n(\bar \rho_n)$ also satisfies~\cref{eq_not_yet_PI}.
To see this note $\mathrm{SYM}_n$ is trace-preserving and completely positive, and hence the data-processing inequality (DPI) for the fidelity~\cite{marco_book} implies 
 \begin{align}
F(\tilde \rho_n , \rho_n)
\overset{\textnormal{perm.~inv.}}{=}F\big( \mathrm{SYM}_n(\bar \rho_n) , \mathrm{SYM}_n(\rho_n) \big)
\overset{\textnormal{DPI}}{\geq} F( \bar \rho_n , \rho_n) 
\overset{\textnormal{\Cshref{eq_not_yet_PI}}}{\geq} \eps \, . \label{eq_fidelity_ds_stein_a}
 \end{align}
Furthermore, we have
\begin{align}
\tilde \rho_n 
=\mathrm{SYM}_n(\bar \rho_n)
\overset{\textnormal{\Cshref{eq_not_yet_PI}}}{\leq}\ee^{yn} \mathrm{SYM}_n(\bar \sigma_n)
\overset{\textnormal{perm.~inv.}}{=} \ee^{yn} \bar \sigma_n \, . \label{eq_op_Stein_PI}
\end{align}
Let $\ket{\theta}_{AE} \in \cH_{A} \otimes \cH_E$ be a purification of $\rho_{A}=\rho$ where $E$ is the purifying system. \Cref{lem_uhlmann_PI} ensures the existence of a permutation-invariant purification $\ket{\Psi_{n-r}}_{A^{n-r} E^{n-r}} \in \cH_{AE}^{\otimes n-r}$ of $\tilde \rho_{n-r}:=\tr_{r}[\tilde \rho_n]$ such that for $d =\dim(\cH)$ we have
\begin{align}
\log | \spr{\Psi_{n-r}}{\theta^{\otimes n-r}}|^2 
\overset{\textnormal{\Cshref{lem_uhlmann_PI}}}&{=} \log F(\tilde \rho_{n-r} , \rho^{\otimes n-r}) \\
\overset{\textnormal{\Cshref{lem_new_fidelity_almost_iid}}}&{\geq} \log F(\tilde \rho_{n} , \rho_{n}) - 2n h\Big(\frac{r}{n}\Big) - 4r \log d\\
\overset{\textnormal{\Cshref{eq_fidelity_ds_stein_a}}}&{\geq} \log \eps - 2n h\Big(\frac{r}{n}\Big) - 4r \log d  \\
&=- \left(\log \frac{1}{\eps} +  2n h\Big(\frac{r}{n}\Big) + 4r \log d \right) =: - \nu_{n}\, . \label{eq_fidelity_stein_fb_0}
\end{align}

We can now apply a variant of the exponential de Finetti theorem~\cite{renner_phd,renner07} stated in~\cref{lem_exp_deFinetti_brandao}, which proves the existence of pure states $\ket{\Phi_{n-r-m}} \in \St^{n-r-m}(\cH_{AE},\ket{\theta}_{AE}^{\otimes n-r-m- r'})$ and $\ket{\Upsilon_{n-r-m}} \in \cH_{AE}^{\otimes n-r -m}$ for any $m, r' \in \N$ such that $m\leq n-r$ and $r'\leq n-r-m$ (which we will choose later) such that
\begin{align}
&\hspace{-40mm}\norm{\tr_{E^{n-r-m}}[\proj{\Upsilon_{n-r-m}}]-\tr_{E^{n-r-m}}[\proj{\Phi_{n-r-m}}]}_1 \nonumber \\
\overset{\textnormal{contractivity~\cite[Thm.~8.16]{wolf_notes}}}&{\leq} \norm{\proj{\Upsilon_{n-r-m}}-\proj{\Phi_{n-r-m}}}_1  \\
\overset{\textnormal{\Cshref{lem_exp_deFinetti_brandao}}}&{\leq} \frac{2 \sqrt{2}}{|\spr{\Psi_{n-r}}{\theta^{\otimes n-r}}|} \ee^{- \frac{m  r'}{2(n-r)}}\\
\overset{\textnormal{\Cshref{eq_fidelity_stein_fb_0}}}&{\leq}2 \sqrt{2} \exp \Big( - \frac{m r'}{2(n-r)} + \frac{\nu_n}{2} \Big)=:2\xi_{n} \, , \label{eq_stein_exp_dF1}
\end{align}
and
\begin{align}
\tr_{E^{n-r-m}}[ \proj{\Upsilon_{n-r-m}}] 
\overset{\textnormal{\Cshref{lem_exp_deFinetti_brandao}}}&{\leq} \frac{1}{|\spr{\Psi_{n-r}}{\theta^{\otimes n-r}}|^2} \tr_{m}[\tilde \rho_{n-r}] \\
\overset{\textnormal{\Cshref{eq_fidelity_stein_fb_0}}}&{\leq} \tr_{m}[\tilde \rho_{n-r}] \ee^{\nu_n}  \\
\overset{\textnormal{\Cshref{{eq_op_Stein_PI}}}}&{\leq} \ee^{y n + \nu_n} \tr_{m }[\bar \sigma_{n-r}] \, . \label{eq_stein_exp_dF2}
\end{align}

By definition of the relative entropy of resource and recalling that the operator logarithm is monotone~\cite[Table~2.2]{Sutter_book}, we find for $\upsilon_{n-r-m}:=\tr_{E^{n-r-m}}[\proj{\Upsilon_{n-r-m}}]$
\begin{align}
\frac{1}{n}  D(\upsilon_{n-r-m} \| \cM_{n-r-m})
\overset{\textnormal{Axiom}~\eqref{axiom_partial_trace}}&{\leq} \frac{1}{n}D\big(\upsilon_{n-r-m} \| \tr_{m }[\bar \sigma_{n-r}] \big) \label{eq_need_ass_here} \\
&= \frac{1}{n}\tr\big[\upsilon_{n-r-m} \big(\log(\upsilon_{n-r-m}) - \log(\tr_{m}[\bar \sigma_{n-r}])   \big)\big] \\
\overset{\textnormal{\Cshref{eq_stein_exp_dF2}}}&{\leq} \frac{1}{n}\log\left(\ee^{yn+\nu_n} \right) +\frac{1}{n}\tr\big[\upsilon_{n-r-m} \big(\log(\tr_{m}[\bar \sigma_{n-r}]) - \log(\tr_{m}[\bar \sigma_{n-r}])   \big)\big]  \\
&=y + \frac{\nu_n}{n} \\
\overset{\textnormal{\Cshref{eq_fidelity_stein_fb_0}}}&{=} y + o(1)  \, , \label{eq_stein_almost_done_step_ds1}
\end{align}
where in the final step, we used~\cref{eq_fidelity_stein_fb_0}, which shows that $\nu_n =o(n)$ since $r=o(n)$.
The continuity of the relative entropy of resource~\cite[Lemma~7]{win16} yields for $\phi_{n-r-m}:=\tr_{E^{n-r-m}}[\proj{\Phi_{n-r-m}}]$
\begin{align}
&\frac{1}{n}  D(\upsilon_{n-r-m} \| \cM_{n-r-m} )
\overset{\textnormal{\Cshref{eq_stein_exp_dF1}}}{\geq} \frac{1}{n}  D(\phi_{n-r-m} \| \cM_{n-r-m}) - f(\xi_{n}) \, , \label{eq_stein_almost_done_step_ds2}
\end{align}
where
\begin{align}
f(\xi_{n}) = \xi_n \log \frac{1}{\lambda_{\min}(\omega)} + \frac{1+\xi_{n}}{n} h\Big(\frac{\xi_{n}}{1+\xi_{n}} \Big) \, , 
\end{align}
where $\lambda_{\min}(\omega)>0$ is a constant given via Axiom~\eqref{axiom_full_rank}.

Next, we choose the parameters $m$ and $r'$, which are functions of $n$, such that
\begin{enumerate}[(a)]
\item $m=o(n)$ and $r'=o(n)$; \label{it_requirement1_sublinear}
\item $\xi_n = o(1)$ where $\xi_n$ is defined in~\cref{eq_stein_exp_dF1}.  \label{it_requirement2_vanishing_xi}
\end{enumerate}
Since $r=o(n)$, the property stated in~\eqref{it_requirement1_sublinear} ensures that we can choose $n$ large enough such that $m\leq n-r$ and $r'\leq n-r-m$.
A possible choice for $m$ and $r'$ is given as follows: let 
\begin{align} \label{eq_gamma_def}
\gamma_n:= \sqrt{\frac{\sqrt{n/r}}{\log(n/r)}} \, ,
\end{align}
where we note that $\gamma_n$ goes to $+\infty$ as $n \to \infty$. 
We next define
\begin{align} \label{eq_choice_params}
m = \sqrt{r n} \qquad \textnormal{and} \qquad  r'= \sqrt{rn} \log\left(\frac{n}{r}\right) \gamma_n \, .
\end{align}
It is straightforward to see that $m=o(n)$ because $r=o(n)$. In addition, we have $r' =o(n)$, because
\begin{align}
 \frac{r'}{n} 
 \overset{\textnormal{\Cshref{eq_choice_params}}}{=} \sqrt{\frac{r}{n}} \log \left( \frac{n}{r} \right) \gamma_n
 \overset{\textnormal{\Cshref{eq_gamma_def}}}{=}  \sqrt{\frac{\log(n/r)}{\sqrt{n/r}}} \, ,
\end{align}
which goes to $0$ as $n \to \infty$. Hence~\eqref{it_requirement1_sublinear} is satisfied.
To see why~\eqref{it_requirement2_vanishing_xi} is fulfilled, we need to convince ourselves that \smash{$\frac{m r'}{n-r}$} grows faster than $\nu_n$. Recall that
\begin{align}
\nu_n 
\overset{\textnormal{\Cshref{eq_fidelity_stein_fb_0}}}&{=}  \log \frac{1}{\eps} +  2n h\Big(\frac{r}{n}\Big) + 4r \log d \, , 
\end{align} 
where the first term is constant in $n$, the second grows as $O(r \log(n/r))$, and the third as $O(r)$. 
From this we see that the second term is the dominant one. Hence, there exists a constant $c>0$ such that $\nu_n  \leq c\, r \log(n/r)$.
Using that $r=o(n)$ we find
\begin{align}
\lim_{n \to \infty} \left( -\frac{m r'}{n-r} + \nu_n \right)
\overset{\textnormal{\Cshref{eq_choice_params}}}&{\leq}\lim_{n \to \infty} \left( - r \log \Big(\frac{n}{r} \Big) \gamma_n + c r \log \Big(\frac{n}{r} \Big) \right) \\
&= \lim_{n \to \infty} r \log \Big(\frac{n}{r} \Big) \left(- \gamma_n +c \right) \\
&= - \infty \, ,
\end{align}
and thus $\xi_n = o(1)$.

Since $D(\phi_{n-r-m} \| \cM_{n-r-m} ) \leq (n-r-m) \log \frac{1}{\lambda_{\min}(\omega)}$, which follows by~\cref{lem_RER_bounded}, we obtain 
\begin{align}
\frac{1}{n} D(\phi_{n-r-m} \| \cM_{n-r-m} )
&\geq \frac{1}{n-r-m} D(\phi_{n-r-m} \| \cM_{n-r-m}) -  \frac{r+m}{n} \log \frac{1}{\lambda_{\min}(\omega)}\\
&= \frac{1}{n-r-m} D(\phi_{n-r-m} \| \cM_{n-r-m}) - o(1) \, . \label{eq_stein_almost_done_step_ds3}
\end{align}
Thus, we find
\begin{align}
\frac{1}{n-r-m} D(\phi_{n-r-m} \| \cM_{n-r-m})\overset{\textnormal{\Cshref{eq_stein_almost_done_step_ds1,eq_stein_almost_done_step_ds2,eq_stein_almost_done_step_ds3}}}{\leq} y + o(1) \, . \label{eq_almost_done_dssssssss}
\end{align}
Furthermore, we have
\begin{align}
\frac{1}{n-r-m} D(\phi_{n-r-m} \| \cM_{n-r-m})
&\geq  \frac{1}{n-r-m}  D\big( \St^{n-r-m}(\cH,\rho^{\otimes n-r-m-r'}) \| \cM_{n-r-m} \big)\\
\overset{\textnormal{\Cshref{lem_limits}}}&{=} \frac{1}{n}  D\big( \St^n(\cH,\rho^{\otimes n-r'}) \| \cM_{n} \big) + o(1) \, . \label{eq_dsssssss1}
\end{align}
Combining the steps above yields
\begin{align}
\frac{1}{n}  D\big( \St^n(\cH,\rho^{\otimes n-r'}) \| \cM_{n} \big)
\overset{\textnormal{\Cshref{eq_relate_Dmax_achievability,eq_DEF_y,eq_almost_done_dssssssss,eq_dsssssss1}}}{\leq}  \frac{1}{n}   D_H^{\eps}\big( \St^n(\cH, \rho^{\otimes n -r}) \| \cM_n \big) + o(1) \, . \label{eq_dddddoooonnneeee}
\end{align}
Taking the limit $n \to \infty$ gives
\begin{align}
D^\reg(\rho \| \cM)  
\overset{\textnormal{\Cshref{lem_regul_almost_iid_irrelevant_Stein}}}{=} \lim_{n \to \infty } \frac{1}{n} D\big( \St^n(\cH,\rho^{\otimes n-r'}) \| \cM_{n} \big) 
\overset{\textnormal{\Cshref{eq_dddddoooonnneeee}}}{\leq}  \lim_{n \to \infty}\frac{1}{n}   D_H^{\eps}\big( \St^n(\cH, \rho^{\otimes n -r}) \| \cM_n \big) \, ,
\end{align}
which proves the achievability part.
\end{proof}

The smooth max-relative entropy and the hypothesis testing relative entropy are closely related by~\cref{eq_buscemi}. Hence, \cref{thm_AGSL} can be rephrased in terms of the smooth max-relative entropy.
\begin{corollary} \label{cor_RGQSL_Dmax}
Let $\eps \in (0,1)$, $\rho \in \St(\cH)$, $\{\cM_n\}_{n \in \N}$ be a family of sets that satisfy the Axioms~\eqref{axiom_convexity}-\eqref{axiom_tensor_product}, and $r=o(n)$. Then
\begin{align}
  \lim_{n \to \infty} \frac{1}{n}  D_{\max}^{\eps}\big(\St^n(\cH, \rho^{\otimes n-r})\| \cM_n \big)
 = D^\reg(\rho \| \cM) \, .
\end{align}
\end{corollary}
\begin{proof}
For any $\rho_n \in  \St^n(\cH, \rho^{\otimes n-r})$ and $\sigma_n \in \cM_n$ we have
\begin{align} \label{eq_cor_step1}
\frac{1}{n} D_{\max}^{\eps}(\rho_n \| \sigma_n)
\overset{\textnormal{\Cshref{eq_buscemi}}}{\leq} \frac{1}{n} D_{H}^{1-\eps^2}(\rho_n \| \sigma_n) + \frac{1}{n} \log \frac{1}{1-\eps^2} 
\end{align}
and hence
\begin{align}
 \lim_{n \to \infty} \frac{1}{n}  D_{\max}^{\eps}\big(\St^n(\cH, \rho^{\otimes n-r})\| \cM_n \big)
 \overset{\textnormal{\Cshref{eq_cor_step1}}}&{\leq}  \lim_{n \to \infty} \frac{1}{n}  D_{H}^{1-\eps^2}\big(\St^n(\cH, \rho^{\otimes n-r})\| \cM_n \big) \\
   \overset{\textnormal{\Cshref{thm_AGSL}}}&{=}  D^\reg(\rho \| \cM)\, .
\end{align}

To see the other direction, note that for any $\rho_n \in  \St^n(\cH, \rho^{\otimes n-r})$, $\sigma_n \in \cM_n$ and $\nu \in (0,1-\eps^2)$ we have
\begin{align} \label{eq_cor_step2}
\frac{1}{n} D_{\max}^{\eps}(\rho_n \| \sigma_n)
\overset{\textnormal{\Cshref{eq_buscemi}}}{\geq} \frac{1}{n} D_{H}^{1-\eps^2-\nu}(\rho_n \| \sigma_n) + \frac{1}{n} \log \frac{1}{1-\eps^2} - \frac{1}{n} \log \frac{4}{\nu^2} 
\end{align}
and therefore
\begin{align}
\lim_{n \to \infty} \frac{1}{n} D_{\max}^{\eps}\big(\St^n(\cH, \rho^{\otimes n-r}) \| \cM_n \big)
 \overset{\textnormal{\Cshref{eq_cor_step2}}}&{\geq}  \lim_{n \to \infty} \frac{1}{n} D_{H}^{1-\eps^2-\nu}\big(\St^n(\cH, \rho^{\otimes n-r})\| \cM_n \big) \\
   \overset{\textnormal{\Cshref{thm_AGSL}}}&{=}  D^\reg(\rho \| \cM)\, .
\end{align}
\end{proof}

\subsection{Robust quantum Stein's lemma}
We next prove a robust version of the quantum Steins's lemma (see~\cref{eq_ASL}), which shows that the relative entropy is not only the optimal error exponent between distinguishing between two iid sources but also for two almost-iid sources.\footnote{This does not follow immediately from~\cref{thm_AGSL} because the set of almost-iid states does not satisfy the Axioms~\eqref{axiom_partial_trace}-\eqref{axiom_tensor_product}.} We first need a technical lemma that ensures that the relative entropy between two almost-iid states asymptotically coincides with the relative entropy of the two iid states.

Before stating and proving the lemma we review some basic properties of almost-iid states which are discussed in more detail in~\cite{MSR26} and will be used in the proofs of~\cref{lem_single_letter_two_almost_iid,thm_ASL}.
By definition of almost-iid states (see~\cref{def_almost_product_state_mixed}), for any $\sigma_{n} \in \St^n(\cH, \sigma^{\otimes n -r})$ there exists a purification $\ket{\theta}_{AE}$ of $\sigma = \sigma_{A}$ and an extension $\sigma_{A^n E^n}$ of $\sigma_n=\sigma_{A^n}$ that can be written as
\begin{align} \label{eq_decomp_psi}
\sigma_{A^n E^n} = \sum_{t,t' \in \cT} \beta_{t,t'} \ket{\Psi_t} \bra{\Psi_{t'}}_{A^n E^n} \, ,
\end{align}
for a family $\{\ket{\Psi_t}_{A^n E^n}\}_{t \in \cT}$ of orthonormal vectors from $\cV(\cH_{AE}^{\otimes n}, \ket{\theta}_{AE}^{\otimes n-r})$ with $\beta_{t,t'} \in \C$ satisfying $\sum_{t \in \cT} \beta_{t,t} =1$ and 
\begin{align} \label{eq_size_T}
\log |\cT| \leq n \, h\Big( \frac{r}{n} \Big) + r \log d_{AE} \, .
\end{align}
Because $\sigma_{A^n E^n}$ given in~\cref{eq_decomp_psi} is a density matrix (and therefore positive semidefinite) and $\{\ket{\Psi_t}_{A^n E^n}\}_{t \in \cT}$ are orthonormal, it follows that $\beta_{t,t} \geq 0$. 
For 
\begin{align} \label{eq_tilde_rho}
\tilde \sigma_{A^n E^n} := \sum_{t \in \cT} \beta_{t,t} \underbrace{\proj{\Psi_t}_{A^n  E^n}}_{=:\tilde \sigma_n^{(t)}}  \, .
\end{align}
It is shown in~\cite[Lemma~4.3]{MSR26} that 
\begin{align} \label{eq_pinching_statement}
\sigma_{A^n} \leq |\cT| \tilde \sigma_{A^n } \, .
\end{align}
The quasi-convexity of the max-divergence~\cite{marco_book} and the structure of $\tilde \sigma_n$ allow us to simplify
\begin{align}
\frac{1}{n} D_{\max}(\tilde \sigma_n  \|  \sigma^{\otimes n} )
\overset{\textnormal{quasi-convexity}}&{\leq} \max_{t \in \cT} \frac{1}{n} D_{\max}(\tilde \sigma^{(t)}_n  \|  \sigma^{\otimes n} ) \\
&= \max_{t \in \cT} \frac{1}{n} D_{\max}(\omega^{(t)}_r  \|  \sigma^{\otimes r} ) \\
&\leq \frac{r}{n} \log \frac{1}{\lambda_{\min}(\sigma)} \label{eq_lam_min_stop} \, ,
\end{align}
where $\omega^{(t)}_r \in \St(\cH^{\otimes r})$ represents the defects in $\sigma^{(t)}_n$.

\begin{lemma} \label{lem_single_letter_two_almost_iid}
Let $\eps \in (0,1)$, $\rho, \sigma \in \St(\cH)$ with $\sigma>0$ and $r=o(n)$. Then
\begin{align}
\lim_{n \to \infty} \frac{1}{n} D\big(\rho^{\otimes n}  \| \St^n(\cH, \sigma^{\otimes n -r}) \big)
 = D(\rho \| \sigma) \, .
\end{align}
\end{lemma}
\begin{proof}
One direction is simple because $\sigma^{\otimes n} \in \St^n(\cH, \sigma^{\otimes n-r})$ and hence
\begin{align}
\lim_{n \to \infty} \frac{1}{n} D\big(\rho^{\otimes n}  \| \St^n(\cH, \sigma^{\otimes n -r}) \big)
 \leq D(\rho \| \sigma) \, .
\end{align}
Thus, it remains to prove the reverse direction.
The operator monotonicity of the logarithm~\cite[Table~2.2]{Sutter_book} implies
\begin{align}
\frac{1}{n} D(\rho^{\otimes n} \| \sigma_n ) 
\overset{\textnormal{\Cshref{eq_pinching_statement}}}&{\geq} \frac{1}{n} D(\rho^{\otimes n} \|  \tilde \sigma_n ) - \frac{1}{n} \log |\cT| \\
 \overset{\textnormal{\cite[Thm.~3.1]{Christandl2017}}}&{\geq}\frac{1}{n} D(\rho^{\otimes n} \|  \sigma^{\otimes n} ) - \frac{1}{n} D_{\max}(\tilde \sigma_n  \|  \sigma^{\otimes n} ) - \frac{1}{n} \log |\cT| \\
 \overset{\textnormal{\Cshref{eq_size_T}}}&{=} D(\rho \|  \sigma) -   \frac{1}{n} D_{\max}(\tilde \sigma_n  \|  \sigma^{\otimes n} )  -o(1) \, ,   \label{eq_ASL_step1}
\end{align}
where the penultimate step uses a triangle inequality for relative entropies~\cite[Theorem~3.1]{Christandl2017}. Putting everything together yields
\begin{align}
\lim_{n \to \infty} \frac{1}{n} D\big(\rho^{\otimes n}  \| \St^n(\cH, \sigma^{\otimes n -r}) \big)
= \lim_{n \to \infty} \min_{\sigma_n \in \St^n(\cH, \sigma^{\otimes n -r})} \frac{1}{n} D(\rho^{\otimes n} \| \sigma_n) 
\overset{\textnormal{\Cshref{eq_ASL_step1,eq_lam_min_stop}}}&{\geq} D(\rho \|  \sigma) \, ,
\end{align}
which completes the proof.
\end{proof}

\begin{theorem}[Robust quantum Stein's lemma] \label{thm_ASL}
Let $\eps \in (0,1)$, $\rho, \sigma \in \St(\cH)$ with $\sigma>0$ and $r_1=o(n)$, $r_2 =o(n)$. Then
\begin{align}
 \lim_{n \to \infty}\frac{1}{n}   D_H^{\eps}\big( \St^n(\cH, \rho^{\otimes n -r_1}) \| \St^n(\cH, \sigma^{\otimes n -r_2})  \big)
 = D(\rho \| \sigma) \, .
\end{align}
\end{theorem}
\begin{proof}
One direction is simple because $\rho^{\otimes n} \in \St^n(\cH, \rho^{\otimes n-r_1})$, $\sigma^{\otimes n} \in \St^n(\cH, \sigma^{\otimes n-r_2})$, and hence 
\begin{align}
 \lim_{n \to \infty}\frac{1}{n}   D_H^{\eps}\big( \St^n(\cH, \rho^{\otimes n -r_1}) \| \St^n(\cH, \sigma^{\otimes n -r_2})  \big)
\leq  \lim_{n \to \infty}\frac{1}{n}   D_H^{\eps}\big( \rho^{\otimes n} \| \sigma^{\otimes n} \big) 
 \overset{\textnormal{\Cshref{eq_SL}}}{=} D(\rho \| \sigma) \, .
\end{align}
To see the other direction, we first recall that the smooth max-relative entropy satisfies a triangle-like inequality of the following form: for $\xi,\eta,\tau \in \St(\cH)$
\begin{align} \label{eq_triangle_smooth_dmax}
D_{\max}^{\eps}(\xi \| \eta) 
\leq D_{\max}(\xi' \| \eta)
\overset{\textnormal{\cite[Thm.~3.1]{Christandl2017}}}{\leq} D_{\max}(\xi' \| \tau) + D_{\max}(\tau \| \eta)
= D_{\max}^{\eps}(\xi \| \tau) + D_{\max}(\tau \| \eta) \, ,
\end{align}
where we used $\xi' \in \argmin_{\bar \xi \in \cB_{\eps}(\xi)} D_{\max}(\bar \xi \| \tau)$. Hence, using the notation from the proof of~\cref{lem_single_letter_two_almost_iid} and
\begin{align} \label{eq_optmizers_before_triangle}
(\rho_n, \sigma_n) \in \arg \min_{(\rho_n',\sigma_n')\in(\St^n(\cH, \rho^{\otimes n -r_1}), \St^n(\cH, \sigma^{\otimes n -r_2}))}  D_H^{\eps}(\rho_n' \| \sigma_n') 
\end{align}
yields
\begin{align}
\frac{1}{n}   D_H^{\eps}\big( \St^n(\cH, \rho^{\otimes n -r_1}) \| \St^n(\cH, \sigma^{\otimes n -r_2})  \big)
\overset{\textnormal{\Cshref{eq_optmizers_before_triangle}}}&{=}\frac{1}{n}   D_H^{\eps}(\rho_n \| \sigma_n) \\
\overset{\textnormal{\Cshref{eq_buscemi}}}&{\geq} \frac{1}{n} D_{\max}^{\sqrt{1-\eps}}(\rho_n \| \sigma_n) - \frac{1}{n} \log \frac{1}{\eps} \\
\overset{\textnormal{\Cshref{eq_pinching_statement,eq_size_T}}}&{\geq} \frac{1}{n} D_{\max}^{\sqrt{1-\eps}}(\rho_n \| \tilde \sigma_n) - o(1) \\
\overset{\textnormal{\Cshref{eq_triangle_smooth_dmax}}}&{\geq}\frac{1}{n}D_{\max}^{\sqrt{1-\eps}}(\rho_n \| \sigma^{\otimes n}) -\frac{1}{n} D_{\max}(\tilde \sigma_n \| \sigma^{\otimes n}) - o(1) \\
\overset{\textnormal{\Cshref{eq_lam_min_stop}}}&{\geq}\frac{1}{n}D_{\max}^{\sqrt{1-\eps}}(\rho_n \| \sigma^{\otimes n})  - o(1) \\
\overset{\textnormal{\Cshref{cor_RGQSL_Dmax}}}&{=} D(\rho \| \sigma)  - o(1) \, ,
\end{align}
where in the final step, we applied~\cref{cor_RGQSL_Dmax} to the set $\cM_n = \{ \sigma^{\otimes n} \}$ (see~\cref{ex_axioms}).
Taking the limit $n \to \infty$ completes the proof.
\end{proof}

We can use the known relations between the hypothesis testing relative entropy and the smooth max-relative entropy to rephrase~\cref{thm_ASL} in terms of the smooth max-relative entropy.
\begin{corollary} \label{cor_AEP_almost_iid}
Let $\eps \in (0,1)$, $\rho, \sigma \in \St(\cH)$ with $\sigma>0$ and $r_1=o(n)$, $r_2 =o(n)$. Then
\begin{align}
 \lim_{n \to \infty}\frac{1}{n}   D^{\eps}_{\max}\big( \St^n(\cH, \rho^{\otimes n -r_1}) \| \St^n(\cH, \sigma^{\otimes n -r_2})  \big)
 = D(\rho \| \sigma) \, .
\end{align}
\end{corollary}
\begin{proof}
For any $\rho_n \in  \St^n(\cH, \rho^{\otimes n-r_1})$ and $\sigma_n \in  \St^n(\cH, \sigma^{\otimes n-r_2})$ we have
\begin{align} \label{eq_cor_step1_1}
\frac{1}{n} D_{\max}^{\eps}(\rho_n \| \sigma_n)
\overset{\textnormal{\Cshref{eq_buscemi}}}{\leq} \frac{1}{n} D_{H}^{1-\eps^2}(\rho_n \| \sigma_n) + \frac{1}{n} \log \frac{1}{1-\eps^2} 
\end{align}
and hence
\begin{align}
 \lim_{n \to \infty} \! \frac{1}{n}  D^{\eps}_{\max}\big( \St^n(\cH, \rho^{\otimes n -r_1}) \| \St^n(\cH, \sigma^{\otimes n -r_2})  \big)\!\!\!\!
 \overset{\textnormal{\Cshref{eq_cor_step1_1}}}&{\leq} \!\!\!\!\! \lim_{n \to \infty}\! \frac{1}{n} D_{H}^{1-\eps^2}\!\!\big( \St^n(\cH, \rho^{\otimes n -r_1}) \| \St^n(\cH, \sigma^{\otimes n -r_2})  \big) \\
   \overset{\textnormal{\Cshref{thm_ASL}}}&{=}  D(\rho \| \sigma)\, .
\end{align}

To see the other direction, note that for any $\rho_n \in  \St^n(\cH, \rho^{\otimes n-r_1})$, $\sigma_n \in \St^n(\cH, \sigma^{\otimes n-r_2})$ and $\nu \in (0,1-\eps^2)$ we have
\begin{align} \label{eq_cor_step2_2}
\frac{1}{n} D_{\max}^{\eps}(\rho_n \| \sigma_n)
\overset{\textnormal{\Cshref{eq_buscemi}}}{\geq} \frac{1}{n} D_{H}^{1-\eps^2-\nu}(\rho_n \| \sigma_n) + \frac{1}{n} \log \frac{1}{1-\eps^2} - \frac{1}{n} \log \frac{4}{\nu^2} 
\end{align}
and therefore
\begin{align}
&\hspace{-30mm}\lim_{n \to \infty} \frac{1}{n} D_{\max}^{\eps}\big( \St^n(\cH, \rho^{\otimes n -r_1}) \| \St^n(\cH, \sigma^{\otimes n -r_2}) \nonumber \\
\overset{\textnormal{\Cshref{eq_cor_step2_2}}}&{\geq}\lim_{n \to \infty} \frac{1}{n} D_{H}^{1-\eps^2-\nu}\big( \St^n(\cH, \rho^{\otimes n -r_1}) \| \St^n(\cH, \sigma^{\otimes n -r_2}) \\
   \overset{\textnormal{\Cshref{thm_ASL}}}&{=}  D(\rho \| \sigma)\, .
\end{align}
\end{proof}
\Cref{thm_AGSL,cor_RGQSL_Dmax,thm_ASL,cor_AEP_almost_iid} can be viewed as a generalized quantum AEP for the sets $\cM_n$ that satisfy the Axioms~\eqref{axiom_convexity}-\eqref{axiom_tensor_product} and the set of almost-iid states. Our results are related to the setting in~\cite{FFF25}; however, they do not follow from their techniques, as the two sets do not both satisfy the polar assumption.\footnote{The polar assumption is explained in~\cite{FFF25} and is known to not be satisfied for the set of separable states, and whether the set of almost-iid states satisfies this assumption is unknown.}

\begin{remark}[Extension to generalized almost-iid states]
The statements of \Cref{thm_AGSL,thm_ASL} remain valid if you consider \emph{generalized} almost-iid states (as defined below~\cref{def_almost_product_state_mixed}), which are not necessarily permutation invariant.\footnotemark
\end{remark}
\footnotetext{By going through the proofs of~\Cref{thm_AGSL,thm_ASL} one can see that permutation invariance of the almost-iid state is not used.}

\section{Asymptotic superadditivity for states with equal marginals} \label{sec_superadditivity}
The continuity statement of the relative entropy of resource with respect to the Wasserstein distance given in~\cref{prop_continuity_ER_Wasserstein} is useful beyond proving the robust generalized quantum Stein's lemma. 
Superadditivity is a natural property of an entanglement measure $E$ stating that for some density matrix $\rho_{A_1 A_2 B_1 B_2} \in \St(\cH_{A_1 A_2 B_1 B_2})$ we have
\begin{align} \label{eq_superadd_fully_general}
E(A_1 A_2 : B_1 B_2)_{\rho} \geq E(A_1 : B_1 )_{\rho} + E(A_2 : B_2 )_{\rho} \, .
\end{align}
The squashed entanglement is superadditive~\cite[Proposition~4]{CW04}, but the entanglement of formation and the relative entropy of entanglement are not~\cite{hastings09,VW01}.

In the following, we show that for states with equal marginals\footnote{Such states are not necessarily permutation invariant.}, the relative entropy of resource (which contains the relative entropy of entanglement) is asymptotically superadditive. 
Before stating the result, we note that the relative entropy of entanglmement is clearly subadditive for such states. In technical terms, let $n \in \N$, $\tau_A \in \St(\cH_A)$, $\rho_{A^n} \in \St(\cH_A^{\otimes n})$ be such that $\rho_{A_i}=\tau_A$ for all $i \in [n]$ and $\{\cM_n\}_{n \in \N}$ be a family of sets that satisfy the Axiom~\eqref{axiom_tensor_product}, then
\begin{align}
\frac{1}{n}D\big(\rho_{A^n}\| \cM_n \big)
= \frac{1}{n} \min_{\sigma_n \in \cM_n} D\big(\rho_{A^n}\| \sigma_n \big)
\overset{\textnormal{Axiom}~\eqref{axiom_tensor_product}}&{\leq} \frac{1}{n}  D\big(\rho_{A^n}\|(\hat \sigma)^{\otimes n} \big)
=D(\tau \| \hat \sigma) 
=D(\tau \| \cM) \, ,
\end{align}
for $\hat \sigma \in \argmin_{\sigma \in \cM} D(\tau \| \sigma)$.
The more interesting statement is that for this scenario a similar bound into the other direction holds when considering the limit $n \to \infty$.
\begin{proposition}[Asymptotic superadditivity] \label{prop_superadditivity}
Let $\tau_A \in \St(\cH_A)$, $\rho_{A^n} \in \St(\cH_A^{\otimes n})$ be such that $\rho_{A_i}=\tau_A$ for all $i \in [n]$, $\{\cM_n\}_{n \in \N}$ be a family of sets that satisfy the Axioms~\eqref{axiom_convexity}-\eqref{axiom_tensor_product}. Then
\begin{align} \label{eq_superadd_result}
\frac{1}{n}  D\big(\rho_{A^n}\| \cM_n \big)
\geq  D^\reg(\tau \| \cM)  - o(1) \, .
\end{align}
\end{proposition}
\begin{proof}
For $\hat \sigma_{n} \in \argmin_{\sigma_n \in \cM_n} D(\rho_{A^n} \| \sigma_n)$ we have
\begin{align}
D\big(\rho_{A^n}\| \cM_n \big)
= D\big(\rho_{A^n}\| \hat \sigma_n \big)
\!\overset{\textnormal{DPI}}{\geq}\! D\big(\mathrm{SYM}_n(\rho_{A^n})\|\mathrm{SYM}_n( \hat \sigma_n) \big)
\!\overset{\textnormal{Axioms~\eqref{axiom_convexity},\eqref{axiom_PI}}}{\geq}\!D\big(\mathrm{SYM}_n(\rho_{A^n})\| \cM_n \big)\, . \label{eq_superadd_0}
\end{align}
A key ingredient in this proof is the exponential de Finetti theorem~\cite{renner_phd,renner07} which states that there exist families \smash{$\{\rho^{(j,r)}_{A^{n-k}}\}_{j \in [J]}$} and $\{\tau_j \}_{j \in [J]}$ with \smash{$\rho^{(j,r)}_{A^{n-k}} \in \St^{n-k}(\cH_A,\tau_j^{\otimes n-k-r})$}, $\tau_j \in \St(\cH_A)$, and a probability distribution $\{ p_j \}_{j \in [J]}$ such that
\begin{align} \label{eq_exp_DF_sum}
\norm{\tr_k[\mathrm{SYM}_n(\rho_{A^n})] - \sum_{j=1}^J p_j \rho^{(j,r)}_{A^{n-k}} }_1 \leq 3 k^d \ee^{- \frac{k(r+1)}{n}}=:\eps_{n,k,r} \, ,
\end{align}
where $J= \poly(n)$.
We replaced the integral in the standard formulation of the exponential de Finetti theorem~\cite{renner_phd,renner07} with a finite sum, which can be done thanks to Carath\'eodory's theorem. This is proven in~\cref{lem_caratheodory}.
Choosing the parameters $k=r=n^{\frac{2}{3}}$ such that $\eps_{n,k,r}=o(1)$, and recalling that the relative entropy of resource is continuous~\cite[Lemma~13]{ludo25} yields
\begin{align}
\frac{1}{n}D\big(\mathrm{SYM}_n(\rho_{A^n})\| \cM_n \big)
\overset{\textnormal{DPI } \& \textnormal{ Axiom}\,\eqref{axiom_partial_trace}}&{\geq} \frac{1}{n}D\big(\tr_{k}[\mathrm{SYM}_n(\rho_{A^n})]\| \cM_{n-k} \big) \\
\overset{\textnormal{\Cshref{eq_exp_DF_sum}}}&{\geq} \frac{1}{n}D\Big( \sum_{j=1}^J p_j \rho^{(j)}_{A^{n-k}}\| \cM_{n-k} \Big) - o(1)  \\
\overset{\textnormal{\Cshref{lem_RER_bounded}}}&{=} \frac{1}{n-k}D\Big( \sum_{j=1}^J p_j \rho^{(j)}_{A^{n-k}}\| \cM_{n-k} \Big) - o(1) \\
&=\frac{1}{n-k}D\Big( \sum_{j=1}^J p_j \rho^{(j)}_{A^{n-k}}\| \hat \sigma_{n-k} \Big) - o(1) \, , \label{eq_superadd_1}
\end{align}
for some $ \hat \sigma_{n-k} \in \cM_{n-k}$.

We next observe that the ``almost-convexity" property of the entropy~\cite[Equation~1]{win16} implies
\begin{align}
\frac{1}{n-k}D\Big( \sum_{j=1}^J p_j \rho^{(j)}_{A^{n-k}}\| \hat \sigma_{n-k} \Big)
&= - \frac{1}{n-k} H\Big( \sum_{j=1}^J p_j \rho^{(j)}_{A^{n-k}} \Big)  - \frac{1}{n-k}\sum_{j=1}^J p_j \tr[\rho^{(j)}_{A^{n-k}} \log \hat \sigma_{n-k}  ] \\
\overset{\textnormal{almost-convexity}}&{\geq} \sum_{j=1}^J p_j \frac{1}{n-k}D\Big(  \rho^{(j)}_{A^{n-k}}\| \hat \sigma_{n-k} \Big) - \frac{\log J}{n-k} \\
&=\sum_{j=1}^J p_j \frac{1}{n-k}D\Big(  \rho^{(j)}_{A^{n-k}}\| \hat \sigma_{n-k} \Big) - o(1) \\
&\geq \sum_{j=1}^J p_j \frac{1}{n-k}D\Big(  \rho^{(j)}_{A^{n-k}}\| \cM_{n-k} \Big) - o(1) \, , \label{eq_superadd_2}
\end{align}
where the penultimate step uses $J=\poly(n)$.

We next observe that
\begin{align}
\sum_{j=1}^J p_j \frac{1}{n-k}D\Big(  \rho^{(j)}_{A^{n-k}}\| \cM_{n-k} \Big)
\overset{\textnormal{\Cshref{cor_asymptotic_continuity_ER}}}&{=} \sum_{j=1}^J p_j \frac{1}{n-k}D\Big(  \tau_j^{\otimes n-k}\| \cM_{n-k}  \Big) - o(1) \\
\overset{\textnormal{\Cshref{eq_def_regul_2}}}&{\geq} \sum_{j=1}^J p_j D^{\reg}(  \tau_j\| \cM ) - o(1) \\
\overset{\textnormal{convexity}}&{\geq}  D^{\reg}\Big( \sum_{j=1}^J p_j \tau_j\| \cM \Big) - o(1) \, , \label{eq_superadd_3}
\end{align}
where the convexity of $\rho \mapsto D^{\reg}(\rho \| \cM)$ is proven in~\cite[Proposition~13]{MHR02}.

Putting everything together yields
\begin{align}
\frac{1}{n}D\big(\rho_{A^n}\| \cM_n \big)
\overset{\textnormal{\Cshref{eq_superadd_0,eq_superadd_1,eq_superadd_2,eq_superadd_3}}}{\geq}  D^{\reg}\Big( \sum_{j=1}^J p_j \tau_j\| \cM \Big) - o(1)\, . \label{eq_superadd_4}
\end{align}
The proof can be completed from~\cref{eq_superadd_4} using the continuity of the regularized relative entropy of resource~\cite[Lemma C.3]{brandao_Stein_10}, together with the observation that \smash{$\sum_{j=1}^J p_j \tau_j$} is close to $\tau$ in trace distance. To see this, note that
\begin{align}
\norm{\tau - \sum_{j=1}^J p_j \tau_j}_1
\overset{\textnormal{triangle}}{\leq} \norm{\tau - \sum_{j=1}^J p_j \tr_{n-k-1}[\rho^{(j)}_{A^{n-k}}]}_1 + \norm{\sum_{j=1}^J p_j \tr_{n-k-1}[\rho^{(j)}_{A^{n-k}}] - \sum_{j=1}^J p_j \tau_j}_1 \label{eq_triang_1}
\end{align}
and since the symmetrizing map $\mathrm{SYM}$ preserves the marginals
\begin{align}
\norm{\tau - \sum_{j=1}^J p_j \tr_{n-k-1}[\rho^{(j)}_{A^{n-k}}]}_1
&=\norm{\tr_{n-1}[\mathrm{SYM}_n(\rho_{A^n})] - \sum_{j=1}^J p_j \tr_{n-k-1}[\rho^{(j)}_{A^{n-k}}]}_1 \\
\overset{\textnormal{contractivity~\cite[Thm.~8.16]{wolf_notes}}}&{\leq} \norm{\tr_{k}[\mathrm{SYM}_n(\rho_{A^n})] - \sum_{j=1}^J p_j \rho^{(j)}_{A^{n-k}}}_1 \\
\overset{\textnormal{\Cshref{eq_exp_DF_sum}}}&{\leq} \eps_{n,k,r} \, . \label{eq_triang_2}
\end{align}
In addition, we have
\begin{align}
\norm{\sum_{j=1}^J p_j \tr_{n-k-1}[\rho^{(j)}_{A^{n-k}}] - \sum_{j=1}^J p_j \tau_j}_1
\overset{\textnormal{triangle}}&{\leq} \sum_{j=1}^J p_j \norm{ \tr_{n-k-1}[\rho^{(j)}_{A^{n-k}}] - \tau_j}_1 \\
\overset{\textnormal{\cite[Prop.~2.6]{MSR26}}}&{\leq} 4 \sqrt{\frac{r}{n-k}} =:\xi_n\, . \label{eq_triang_3}
\end{align}
For our choice of parameters $k=r=n^{\frac{2}{3}}$ we find $\eps_{n,k,r}=o(1)$ and $\xi_n = o(1)$ and thus 
\begin{align}
\norm{\tau - \sum_{j=1}^J p_j \tau_j}_1
\overset{\textnormal{\Cshref{eq_triang_1,eq_triang_2,eq_triang_3}}}{\leq}  \eps_{n,k,r} + \xi_n = o(1)\, . \label{eq_triang_final}
\end{align}
Hence, the continuity of the regularized relative entropy of resource~\cite[Lemma C.3]{brandao_Stein_10} implies
\begin{align}
D^{\reg}\Big( \sum_{j=1}^J p_j \tau_j\| \cM \Big) 
\overset{\textnormal{\Cshref{eq_triang_final}}}{\geq} D^{\reg}(\tau \| \cM) - o(1) \, . \label{eq_superadd_5}
\end{align}
Combining~\cref{eq_superadd_4,eq_superadd_5} completes the proof.
\end{proof}

For $\cM_{n} = \mathrm{SEP}(A^n:B^n)$,~\cref{prop_superadditivity} shows that for any $\rho_{A^n B^n} \in \St(\cH_{AB}^{\otimes n})$ such that $\rho_{A_i B_i} = \rho_{A_j B_j}$ for all $i,j \in [n]$ we have
\begin{align} \label{eq_our_subadd}
  \liminf_{n \to \infty} \frac{1}{n}  D\big(\rho_{A^n B^n}\| \mathrm{SEP} \big)
 \geq D^\reg(\rho_{A_i B_i} \| \mathrm{SEP}) \, .
\end{align}
This can be contrasted with the superadditivity result from~\cite{piani09}, which shows that 
\begin{align} \label{eq_piani}
  \liminf_{n \to \infty} \frac{1}{n}  D\big(\rho_{A^n B^n}\| \mathrm{SEP}\big)
 \geq  D_{\mathbb{M}}(\rho_{A_i B_i} \| \mathrm{SEP}) \, ,
\end{align}
where $D_\mathbb{M}$ denotes the measured relative entropy with measurements restricted to $\mathrm{SEP}$ or $\mathrm{LOCC}$. \Cref{eq_our_subadd} is stronger than~\cref{eq_piani} because $D^\reg(\rho_{A_i B_i} \| \mathrm{SEP}) \geq  D_{\mathbb{M}}(\rho_{A_i B_i} \| \mathrm{SEP})$ as ensured by~\cite{piani09}.
\paragraph{Acknowledgements} We thank Fernando Brand\~ao for his talk~\cite{brandao_talk}.
We also thank Giacomo De Palma, Filippo Girardi, Ludovico Lami, and Ryuji Takagi for discussions on the robustness of entanglement measures with respect to the quantum Wasserstein distance. We further acknowledge discussions with Nilanjana Datta, Christophe Piveteau, Joe Renes, Lukas Schmitt, and Marco Tomamichel on the generalized quantum Stein's lemma. GM and RR acknowledge support from the NCCR SwissMAP, the ETH Zurich Quantum Center, the SNSF project No.~20QU-1 225171, and the CHIST-ERA project MoDIC.




\appendix
\section{Technical statements}
A sequence $\{a_n \}_{n \in \N}$ is subadditive if it satisfies $a_{k+m} \leq a_k + a_m$ for all $k,m \in \N$.
\begin{lemma}[{Fekete's subadditivity lemma~\cite{fekete23}}] \label{lem_fekete}
Let $ \{ a_n \}_{n \in \N}$ be a subadditive sequence. Then the limit $\lim_{n \to \infty} \frac{a_n}{n}$ exists and is equal to $\inf_{n \in \N} \frac{a_n}{n}$.
\end{lemma}
This lemma ensures that the regularized relative entropy of resource 
\begin{align} \label{eq_regul_exists}
D^\reg(\rho \| \cM)
:=\lim_{n \to \infty} \frac{1}{n} D(\rho^{\otimes n} \| \cM_n)
=\inf_{n \in \N}\frac{1}{n} D(\rho^{\otimes n} \| \cM_n) 
\end{align}
is well-defined. To see this, note that
\begin{align} \label{eq_justification_fekete_ER}
D(\rho^{\otimes k+m} \| \cM_{k+m}) \leq D(\rho^{\otimes k} \| \cM_k)+D(\rho^{\otimes m} \| \cM_m) \quad \textnormal{for all} \, k,m \in \N \, .
\end{align}
This is true since $\sigma_k \in \cM_k$ and $\sigma_m \in \cM_m$ implies $\sigma_k \otimes \sigma_m \in \cM_{k+m}$ (by Axiom~\eqref{axiom_tensor_product}) and the relative entropy is additive under tensor products~\cite{marco_book}.

\begin{lemma} \label{lem_generalized_almost_iid_satisfies_axioms}
For $n,r\in \N$ with $r\leq n$ and $ \St(\cH) \ni \omega>0$, the set of generalized almost-iid states $\cM_n = \bar \St^n(\cH,\omega^{\otimes n-r})$ satisfies  Axioms~\eqref{axiom_Giulia_1}-\eqref{axiom_Giulia_3}.
\end{lemma}
\begin{proof}
It follows from~\cite[Remark 2.4 (g) and (h)]{MSR26} that Axiom~\eqref{axiom_Giulia_1} is satisfied. To see that Axiom~\eqref{axiom_Giulia_3} is satisfied, we need to show that for any $\rho_{A^n} \in \cM_n$ we have \smash{$\cR^{(\omega)}_i(\rho_{A^n}) = \omega \otimes \tr_{i}[\rho_{A^n}] \in \cM_n$}. We recall that~\cite[Remark 2.4 (d)]{MSR26} ensures that $\bar \rho_{A^{n-1}}:=\tr_{i}[\rho_{A^n}] \in \bar \St^{n-1}(\cH,\omega^{\otimes n-1-r})$.\footnote{Following the proof of~\cite[Lemma~B.3]{MSR26} we see that the statement remains valid for generalized almost-iid states.} Then, because of~\cite[Remark 2.4 (f)]{MSR26}, we directly find that $\omega \otimes \tr_{i}[\rho_{A^n}] \in \cM_n$.

\end{proof}

\begin{lemma}[DPI for smooth max-relative entropy] \label{lem_DPI_Dmax_eps}
Let $\eps \in (0,1)$, $\rho,\sigma \in \St(\cH)$, and $\cE$ be a trace-preserving completely positive map. Then
\begin{align}
D^\eps_{\max}\big( \cE(\rho) \| \cE(\sigma) \big) \leq D^\eps_{\max}(\rho \| \sigma) \, .
\end{align}
\end{lemma}
\begin{proof}
Let $\bar \rho \in \cB_{\eps}(\rho)$ be such that $D^\eps_{\max}(\rho \| \sigma) = D_{\max}(\bar \rho \| \sigma)$. The monotonicity of the purified distance under quantum channels~\cite[Equation~3.49]{marco_book} gives
\begin{align} \label{eq_ball_DPI_Dmax}
P\big(\cE(\bar \rho), \cE(\rho) \big) \leq P(\bar \rho , \rho) \leq \eps \, .
\end{align}
Hence, we have
\begin{align}
D^\eps_{\max}\big( \cE(\rho) \| \cE(\sigma) \big)
&=\min_{\tilde \rho \in \cB_{\eps}(\cE(\rho))}  D_{\max}\big(\tilde \rho \| \cE(\sigma) \big) \\
\overset{\textnormal{\Cshref{eq_ball_DPI_Dmax}}}&{\leq} D_{\max}\big(\cE(\bar \rho) \| \cE(\sigma) \big) \\
\overset{\textnormal{\cite[Prop~4.7]{marco_book}}}&{\leq} D_{\max}(\bar \rho \| \sigma) \\
&=D^\eps_{\max}(\rho \| \sigma) \, .
\end{align}
\end{proof}

We next recall two modifications of Uhlmann's theorem~\cite{Uhl76}.
\begin{lemma}\label{lem_uhlmann_mixed}
Let $\rho_{AB} \in \St(\cH_{AB})$ and $\sigma_A \in \St(\cH_A)$. Then there exists an extension $\sigma_{AB} \in \St(\cH_{AB})$ of $\sigma_{A}$ such that
\begin{align}
F(\rho_A , \sigma_A) = F(\rho_{AB} , \sigma_{AB}) \, .
\end{align}
\end{lemma}
\begin{proof}
Let $\ket{\psi}_{ABR}$ be an arbitrary purification of $\rho_{AB}$ and $\ket{\phi}_{ABR}$ be a purification of $\sigma_A$ such that $F(\rho_{A} , \sigma_{A}) = |\spr{\psi}{\phi}|^2 $, which exists by Uhlmann's theorem~\cite{Uhl76}. For $\sigma_{AB}:= \tr_R [\proj{\phi}_{ABR}]$ we have
\begin{align}
F(\rho_{AB} , \sigma_{AB}) 
\overset{\textnormal{DPI}}{\leq} F(\rho_{A} , \sigma_{A}) 
\overset{\textnormal{Uhlmann's theorem}}{=}  |\spr{\psi}{\phi}|^2 
 \overset{\textnormal{DPI}}{\leq} F(\rho_{AB} , \sigma_{AB}) \, .
\end{align}
\end{proof}

\begin{lemma}[{\cite[Lemma~III.4]{brandao_Stein_10}}] \label{lem_uhlmann_PI}
Let $\rho \in \St(\cH)$ and $\rho_n \in \St(\cH^{\otimes n})$ be permutation invariant. Then there exists a purification $\ket{\theta} \in \cH \otimes \cH$ of $\rho$ and a permutation-invariant purification $\ket{\Psi_n} \in (\cH \otimes \cH)^{\otimes n}$ of $\rho_n$ such that
\begin{align}
 |\spr{\Psi_n}{\theta^{\otimes n}}|^2 = F(\rho_n,\rho^{\otimes n}) \, .
\end{align}
\end{lemma}

\begin{lemma} \label{lem_new_fidelity_almost_iid}
Let $n,r \in \N$ be such that $n-r  \in \N$, $\rho \in \St(\cH)$, $\rho_n \in \St^n(\cH, \rho^{\otimes n-r})$, $\omega_n \in \St(\cH^{\otimes n})$ be permutation invariant and $d=\dim \cH$. Then
\begin{align}
F(\omega_n, \rho_n) \leq F\big(\omega_{n-r},\rho^{\otimes n-r}\big) |\cT|^2 \, ,
\end{align} 
for $\log |\cT| \leq n h(r/n) + 2r \log d$.
\end{lemma}
\begin{proof}
To simplify the notation, let $\cH = \cH_A$.
Because $\rho_n = \rho_{A^n}$ is an almost-iid state, there exist a purification $\ket{\theta}_{AE}$ of $\rho=\rho_A$ and an extension $\rho_{A^n E^n}$ such that
\begin{align} \label{eq_pinching_new_lemma}
\rho_{A^n E^n} \overset{\textnormal{\cite[Lemma~4.3]{MSR26}}}{\leq} |\cT| \tilde \rho_{A^n E^n} = |\cT| \sum_{t \in \cT} p_{t} \proj{\Psi_t}_{A^n E^n} \, ,
\end{align}
for a family $\{\ket{\Psi_t}_{A^n E^n}\}_{t \in \cT}$ of orthonormal vectors from $\cV(\cH_{AE}^{\otimes n}, \ket{\theta}_{AE}^{\otimes n-r})$ with $p_t \in \R$ satisfying $\sum_{t \in \cT} p_{t} =1$ and
\begin{align} \label{eq_size_T2}
\log |\cT| 
\overset{\textnormal{\cite[Remark~2.4]{MSR26}}}{\leq} n h\Big( \frac{r}{n}\Big) + 2r \log d \, ,
\end{align}
for $d=\dim \cH$.
Let  $\omega_{A^n E^n}$ be an extension of $\omega_n = \omega_{A^n}$ such that
\begin{align}
F(\omega_{A^n} , \rho_{A^n}) 
\overset{\textnormal{\Cshref{lem_uhlmann_mixed}}}{=}F(\omega_{A^n E^n} , \rho_{A^n E^n}) 
\overset{\textnormal{\Cshref{eq_pinching_new_lemma}}}{\leq} F(\omega_{A^n E^n} , \tilde \rho_{A^n E^n}) |\cT| \, , \label{eq_step_1_new_uhlmann}
\end{align}
where the final step uses the monotonicity of the matrix square root~\cite[Table~2.2]{Sutter_book}.
Let 
\begin{align} \label{eq_purification_Uhl}
\ket{\xi}_{A^n E^n R} = \sum_{t \in \cT} \sqrt{p_{t}} \ket{\Psi_t} \otimes \ket{t}_R \, 
\end{align}
be a purification of $\tilde \rho_{A^n E^n}$ and let $\ket{\Omega}_{A^n E^n R}$ be a purification of $\omega_{A^n E^n}$. By Uhlmann's theorem~\cite{Uhl76} we thus find
\begin{align}
\sqrt{F(\omega_{A^n E^n} , \tilde \rho_{A^n E^n}) }
&= \sup_{U_R} | \bra{\Omega}(\id_{A^n E^n} \otimes U_R)\ket{\xi}| \\
\overset{\textnormal{\Cshref{eq_purification_Uhl}}}&{\leq} \sup_{U_R}  \sum_{t \in \cT} \sqrt{p_{t}} |\bra{\Omega}(\id_{A^n E^n} \otimes U_R) \ket{\Psi_t} \otimes \ket{t} |   \\
& \leq \max_{t \in \cT} \sup_{U_R} |\bra{\Omega}(\id_{A^n E^n} \otimes U_R) \ket{\Psi_t} \otimes \ket{t} | \sum_{t \in \cT} \sqrt{p_{t}}
\\
\overset{\textnormal{Cauchy-Schwarz}}&{\leq} \sqrt{|\cT|} \max_{t \in \cT} \sup_{U_R} |\bra{\Omega}(\id_{A^n E^n} \otimes U_R) \ket{\Psi_t} \otimes \ket{t} | \\
& =\sqrt{|\cT|}  \max_{t \in \cT} \sqrt{F(\omega_{A^n E^n} , \proj{\Psi_t}_{A^n E^n})} \, , \label{eq_step_2_new_uhlmann}
\end{align}
where the final step also uses Uhlmann's theorem.
For any $t \in \cT$ we can utilize the structure of $\ket{\Psi_t}_{A^n E^n}$ to write
\begin{align}
F(\omega_{A^n E^n} , \proj{\Psi_t}_{A^n E^n})
\overset{\textnormal{DPI}}&{\leq}F(\omega_{A^n} , \tr_{E^n}[\proj{\Psi_t}_{A^n E^n}]) \\
\overset{\textnormal{DPI $\&$ permutation invariance}}&{\leq} F\Big(\omega_{n-r} , \rho^{\otimes n-r}\Big) \, .  \label{eq_step_3_new_uhlmann}
\end{align}
Combining~\cref{eq_step_1_new_uhlmann,eq_step_2_new_uhlmann,eq_step_3_new_uhlmann} proves the assertion.
\end{proof}

The next lemma is a version of the exponential de Finetti theorem~\cite{renner_phd,renner07}. This result is crucial in the proof of the generalized quantum Stein's lemma as it relates permutation-invariant states to almost-iid states.
\begin{lemma}[{\cite[Lemma~III.5]{brandao_Stein_10}}] \label{lem_exp_deFinetti_brandao}
Let $n,m,r \in \N$ such that $m \leq n$ and $r \leq n-m$, $\ket{\Psi_n} \in \cH^{\otimes n}$ be a permutation-invariant state, and $\ket{\theta} \in \cH$. Then there exist states $\ket{\Upsilon_{n-m}} \in \cH^{\otimes n-m}$ and $\ket{\Phi_{n-m}} \in \mathrm{Sym}(\cH^{\otimes n-m},\ket{\theta}^{\otimes n-m-r})$ such that
\begin{align} \label{eq_fernando_dF_1}
\proj{\Upsilon_{n-m}} \leq \frac{1}{|\spr{\Psi_n}{\theta^{\otimes n}}|^2} \tr_{m}[\proj{\Psi_n}]
\end{align}
and
\begin{align} \label{eq_fernando_dF_2}
\norm{\proj{\Upsilon_{n-m}}-\proj{\Phi_{n-m}}}_1 \leq \frac{2 \sqrt{2}}{|\spr{\Psi_n}{\theta^{\otimes n}}|} \ee^{-\frac{mr}{2n}} \, .
\end{align}
\end{lemma}

\begin{lemma} \label{lem_RER_bounded}
Let $\rho_n \in \St(\cH^{\otimes n})$ and $\{\cM_n\}_{n \in \N}$ be a family of sets that satisfy the Axioms~\eqref{axiom_convexity}-\eqref{axiom_tensor_product}. Then
\begin{align}
\frac{1}{n} D(\rho_n \| \cM_n) \leq \log \frac{1}{\lambda_{\min}(\omega)} \, .
\end{align}
\end{lemma}
\begin{proof}
Using the fact that the logarithm is operator monotone~\cite[Table~2.2]{Sutter_book} we find
\begin{align}
\frac{1}{n} D(\rho_n \| \cM_n) 
\overset{\textnormal{Axioms~\eqref{axiom_full_rank} \& \eqref{axiom_tensor_product}}}&{\leq}\frac{1}{n} D(\rho_n \| \omega^{\otimes n}) \\
&\leq D\big(\rho_n \| \lambda_{\min}(\omega^{\otimes n}) \id_n\big) \\
&= - \frac{1}{n} H(\rho_n) - \frac{1}{n} \tr[\rho_n \log (\lambda_{\min}(\omega^{\otimes n}) \id_n)] \\
&\leq - \log \lambda_{\min}(\omega) \, ,
\end{align}
where the second step uses $\omega^{\otimes n} \geq \lambda_{\min}(\omega^{\otimes n}) \id_n$ and the final step uses the nonnegativity of the entropy and $\lambda_{\min}(\omega^{\otimes n}) = \lambda_{\min}(\omega)^n$.
\end{proof}

\begin{lemma} \label{lem_limits}
Let $f:\N \to \R$ be such that $\lim_{n \to \infty} f(n) = \hat f <\infty$ exists. For $s(n) \in \N$ such that $\lim_{n\to \infty} s(n) = \infty$, we have
\begin{align}
\lim_{n\to \infty} | f(n)-f\big(s(n) \big) | = 0 \, .
\end{align}
\end{lemma}
\begin{proof}
Let $\eps>0$. Since $\lim_{n \to \infty} f(n) = \hat f$ there exists $\hat s \in \N$ such that for all $s\geq \hat s$ we have $|\hat f - f(s(n))| < \eps$. Since $\lim_{n \to \infty} s(n) = \infty$ there exists $\hat n \in \N$ such that for all $n \geq \hat n$ we have $s(n) \geq \hat s$. Hence, for all $n \geq \hat n$ it is true that $|\hat f - f(s(n))| < \eps$ and therefore $\lim_{n \to \infty} f(s(n)) = \hat f$.
\end{proof}

\begin{lemma}[Exponential de Finetti theorem with finite sum] \label{lem_caratheodory} Let $n,k,r\in\N$ such that $k+r \leq n$, and $d=\dim(\cH)$. Let $\rho_n\in \St(\cH^{\otimes n})$ be a permutation-invariant state. Then, there exist families \smash{$\{\rho^{(j,r)}_{{n-k}}\}_{j \in [J]}$} and $\{\tau_j \}_{j \in [J]}$ with \smash{$\rho^{(j,r)}_{{n-k}} \in \St^{n-k}(\cH,\tau_j^{\otimes n-k-r})$}, $\tau_j \in \St(\cH)$, and a probability distribution $\{ p_j \}_{j \in [J]}$ such that
\begin{align} 
\norm{\tr_k[\rho_{n}] - \sum_{j=1}^J p_j \rho^{(j,r)}_{{n-k}} }_1 \leq 3 k^d \ee^{- \frac{k(r+1)}{n}} \, ,
\end{align}
where $J=\binom{n-k+d^2 -1}{n-k}^2 = \poly(n)$.
\end{lemma}
\begin{proof}
For any $n\in \N $ and $\ket{\theta}\in \cH$, let $ \mathrm{Sym}^{n}(\cH, \ket{\theta}^{\otimes n-r}):= \mathrm{Sym}^n(\cH)\cap  \mathrm{span}\,\cV(\cH^{\otimes n}, \ket{\theta}^{\otimes n-r}) $. As shown in~\cite[Lemma 4.2.2.]{renner_phd}, there exists a symmetric purification $\ket{\Phi^{(n)}} \in \mathrm{Sym}^n(\cH \otimes \cH)$ of $\rho_{n}$. The exponential de Finetti theorem~\cite{renner_phd,renner07} then states that there exists a probability measure $\nu$ on the unit sphere $\cB(\cH\otimes \cH)$ and a family \smash{$\{\ket{\Psi^{(n-k)}_{r,\theta}}\}_{\theta}$} of states such that \smash{$\ket{\Psi^{(n-k)}_{r,\theta}} \in \mathrm{Sym}^{n-k}(\cH\otimes \cH, \ket{\theta}^{\otimes n-k-r})$} and
\begin{align} \label{lem_eq-deFinetti-1}
\norm{\tr_{k}[\proj{\Phi^{(n)}}] - \int \proj{\Psi^{(n-k)}_{r,\theta}}  \nu(\di \theta)}_1 \leq 3 k^d \ee^{-\frac{k(r+1)}{n+k}} \, .
\end{align}
Now, consider the set $B^{n-k}_r$ defined by 
\begin{align} \label{eq_setB}
    B^{n-k}_r :=&\, \Big\{\proj{\Psi}\,:\,\exists\, \ket{\theta} \in \cB(\cH\otimes \cH)\ \mathrm{s.t.} \ \ket{\Psi} \in \mathrm{Sym}^{n-k}(\cH\otimes \cH, \ket{\theta}^{\otimes n-k-r}) \Big\} \,.
\end{align}
It turns out (and will be shown below) that $B^{n-k}_r $ is a compact set. Therefore, and since $\cH$ is finite-dimensional, it follows that the convex hull \smash{$\mathrm{conv}(B^{n-k}_r)$} is also compact~\cite[Corollary~2.4]{barvinok2002course}, and consequently~\cite[Proposition~1.2.12]{analysis_banach}
\begin{align}
    \sigma_{n-k} := \int \proj{\Psi^{(n-k)}_{r,\theta}}  \nu(\di \theta)\,\in\,\overline{\mathrm{conv}(B^{n-k}_r)} = \mathrm{conv}(B^{n-k}_r)\,. \label{eq_caratheo1}
\end{align}
On the other hand, we have $B^{n-k}_r \subseteq \mathrm{Herm}(\mathrm{Sym}^{n-k}(\cH \otimes \cH))$ by definition, which further implies that $\mathrm{conv}(B^{n-k}_r) \subseteq \mathrm{Herm}(\mathrm{Sym}^{n-k}(\cH \otimes \cH)) $. For \smash{$D := \dim(\mathrm{Sym}^{n-k}(\cH \otimes \cH)) = \binom{n-k+d^2-1}{n-k}$}, it follows that the space of Hermitian operators on $\mathrm{Sym}^{n-k}(\cH \otimes \cH)$ with unit trace has \emph{real} dimension $D^2-1$. As a consequence, Carath\'{e}odory's theorem then states that $\sigma_{n-k}$ is the convex sum of at most $J =D^2$ elements of $B^{n-k}_r$. This means that
\begin{align} \label{lem_eq-carath}
\sigma_{n-k} \,=\, \sum_{j=1}^J p_j \proj{\tilde{\Psi}^{(n-k)}_{r,\theta_j}}  \quad \textnormal{for} \quad  J = \binom{n-k+d^2-1}{n-k}^2\,,
\end{align}
where $\{ p_j \}_{j \in [J]}$ defines a probability distribution and $\ket{\tilde{\Psi}^{(n-k)}_{r,\theta_j}} \in \mathrm{Sym}^{n-k}(\cH\otimes \cH, \ket{\theta_j}^{\otimes n-k-r})$ for some $\ket{\theta_j} \in \cB(\cH\otimes \cH)$ for all $j\in [J]$. By taking the partial trace over the second copy $\cH$, and due to the contractivity of the trace distance~\cite[Theorem~8.16]{wolf_notes}, we find 
\begin{align}
\norm{\tr_k[\rho_{n}] - \sum_{j=1}^J p_j \rho^{(j,r)}_{{n-k}} }_1 
&\leq \norm{\tr_{k}[\proj{\Phi^{(n)}}] - \sum_{j=1}^J p_j \proj{\tilde{\Psi}^{(n-k)}_{r,\theta_j}} }_1 \\ \overset{\textnormal{\Cshref{lem_eq-deFinetti-1,eq_caratheo1,lem_eq-carath}}}&{\leq} 3 k^d \ee^{-\frac{k(r+1)}{n+k}} \,,
\end{align}
where we have defined $\rho^{(j,r)}_{{n-k}} := \tr_{\cH} \big[\proj{\tilde{\Psi}^{(n-k)}_{r,\theta_j}}\big]$, implying that $\rho^{(j,r)}_{{n-k}} \in \St^{n-k}(\cH,\tau_j^{\otimes n-k-r})$ for $\tau_j := \tr_{\cH}[\proj{\theta_j}]\in \St(\cH)$ for all $j\in[J]$.

It remains to be shown that the set $B^{n-k}_r$ defined in~\cref{eq_setB} is compact. To simplify the notation, let $\cK := \cH \otimes \cH$ and $\ell := n-k$. For any $\ket{\theta}\in \cK$, let $\{\ket{x}\}_{x\in\cX}$ be an orthonormal basis of $\cK$ with $\ket{\bar{x}} = \ket{\theta}$ for some $\bar{x}\in \cX$. For any type $Q \in \cP_\ell(\cX)$ on $\cX$ with denominator $\ell$, we define the normalized state
\begin{align} \label{eq_big_Theta}
    \ket{\Theta^Q_\ell}\,:=\,\frac{1}{\sqrt{|\cT_\ell^Q|}} \sum_{(x_1,\dots,x_\ell) \in \cT_\ell^Q} \ket{x_1}\otimes \cdots \otimes \ket{x_\ell}\,,
\end{align}
where $\cT_\ell^Q$ denotes the type class of $Q$. In~\cite[Lemma 4.1.5.]{renner_phd}, it was shown that $\{\ket{\Theta^Q_\ell}\}_{Q\,:\, \ell Q(\bar{x}) \geq \ell-r}$ is an orthonormal basis of \smash{$ \mathrm{Sym}^{\ell}(\cK, \ket{\theta}^{\otimes \ell-r})$}. This implies that
\begin{align} \label{lem_eq-sym-proj}
    P_{\ell,r}^{\ket{\theta}} \,:=\, \sum_{{Q\,:\, \ell Q(\bar{x}) \geq \ell-r}} \proj{\Theta^Q_\ell}
\end{align}
is an orthogonal projector onto \smash{$ \mathrm{Sym}^{\ell}(\cK, \ket{\theta}^{\otimes \ell-r})$}. Next, we argue that there exists a mapping \smash{$\ket{\theta}\mapsto P_{\ell,r}^{\ket{\theta}}$} that is continuous. To this end, it suffices to construct a continuous mapping $\ket{\theta}\mapsto \ket{\Theta^Q_\ell}$. For this purpose, let $\ket{\xi}\in \cK$ denote some fixed (reference) state and let \smash{$\{\ket{x}\}_{x\in\cX}$} be an orthonormal basis of $\cK$ corresponding to $\ket{\xi}$ as defined above. Furthermore, let \smash{$Q \in \cP_\ell(\cX)$} and let \smash{$ \ket{\Xi^Q_\ell}$} be the normalized state defined in~\cref{eq_big_Theta} corresponding to $\ket{\xi}$. Now, let $\ket{\theta} \in \cK$ be arbitrary. The idea is to construct a unitary that maps the orthonormal basis \smash{$\{\ket{x}\}_{x\in\cX}$} corresponding to $\ket{\xi}$ to a new orthonormal basis corresponding to $\ket{\theta}$ in a continuous way. To do so, first note that since \smash{$ \mathrm{Sym}^{\ell}(\cK, \ket{\theta}^{\otimes \ell-r}) = \mathrm{Sym}^{\ell}(\cK, \ket{\theta'}^{\otimes \ell-r})$} for any $\ket{\theta'} = \ee^{\ci \alpha} \ket{\theta}$ with $\alpha \in \R$, we may assume that $\spr{\xi}{\theta} \geq 0 $ without loss of generality. Then, we define two orthonormal states
\begin{align} \label{lem_onb-elements}
   \ket{e_1} :=  \ket{\xi}\,,\quad \ket{e^\theta_2} := \frac{\ket{\theta} - \spr{\xi}{\theta}\ket{\xi}}{\sqrt{1-|\spr{\xi}{\theta}|^2}} \quad \implies \quad \spr{e_1}{e_1} = \spr{e^\theta_2}{e^\theta_2} = 1\,,\ \spr{e_1}{e^\theta_2} = 0\,,
\end{align}
and, setting $\cos(\varphi) := \spr{\xi}{\theta}$, we construct a unitary
\begin{align} \label{lem_unitary-def}
U_{\ket{\theta}} := \underbrace{\id - \proj{e_1} - \proj{e^\theta_2}}_{=:\, P^\theta_\perp} \,+\, \underbrace{\cos(\varphi)\big(\proj{e_1} + \proj{e^\theta_2}\big) + \sin(\varphi)\big(\ket{e^\theta_2}\!\bra{e_1} - \ket{e_1}\!\bra{e^\theta_2} \big)}_{=: \, U^\theta_{2D}}\,.
\end{align}
Here, $P^\theta_\perp$ denotes a projector onto the space $\mathrm{span}\{\ket{\xi},\ket{\theta}\}^\perp$ and $U^\theta_{2D}$ corresponds to a rotation in the two-dimensional space $\mathrm{span}\{\ket{\xi},\ket{\theta}\}$. In particular, $U_{\ket{\theta}}$ is a unitary since $\smash{U_{\ket{\theta}}^\dagger U_{\ket{\theta}} = U_{\ket{\theta}} U_{\ket{\theta}}^\dagger} = P^\theta_\perp + \id_{2D}  = \id$, and it satisfies $U_{\ket{\theta}} \ket{\xi} = \ket{\theta}$. By further inserting~\cref{lem_onb-elements} into~\cref{lem_unitary-def}, the expression for the unitary $U_{\ket{\theta}}$ simplifies (after a straightforward calculation) to
\begin{align}
U_{\ket{\theta}} = \id +\big(\ket{\theta}\!\bra{\xi} -\ket{\xi}\!\bra{\theta}  \big) + \frac{1}{1+\spr{\xi}{\theta}}\big(\ket{\theta}\!\bra{\xi} -\ket{\xi}\!\bra{\theta}  \big)^2\,.
\end{align}
In particular, $U_{\ket{\theta}}$ is now well-defined for all $\ket{\theta}$ such that $\spr{\xi}{\theta} \geq 0$, and the map $\ket{\theta}\mapsto U_{\ket{\theta}}$ is clearly continuous since complex conjugation and the inner product are continuous operations. With that, we define the map
\begin{align}
 \ket{\theta}\,\mapsto \,\ket{\Theta^Q_\ell} := U_{\ket{\theta}}^{\otimes \ell}  \ket{\Xi^Q_\ell} \,,
\end{align}
where \smash{$\ket{\Theta^Q_\ell} \in  \mathrm{Sym}^{\ell}(\cK, \ket{\theta}^{\otimes \ell-r}) $} since \smash{$U_{\ket{\theta}} \ket{\xi} = \ket{\theta}$}. Taking the tensor product and applying a linear map to a vector are continuous operations, implying that the map \smash{$ \ket{\theta}\mapsto \ket{\Theta^Q_\ell}$} is continuous as desired. Finally, note that \smash{$\{\ket{\Theta^Q_\ell}\}_{Q\,:\, \ell Q(\bar{x}) \geq \ell-r}$} is an orthonormal basis of \smash{$ \mathrm{Sym}^{\ell}(\cK, \ket{\theta}^{\otimes \ell-r})$} because \smash{$U_{\ket{\theta}}$} is unitary. This shows that $ P_{\ell,r}^{\ket{\theta}} $, as defined in~\cref{lem_eq-sym-proj}, is an orthogonal projector onto \smash{$ \mathrm{Sym}^{\ell}(\cK, \ket{\theta}^{\otimes \ell-r})$}, and the resulting map \smash{$\ket{\theta}\mapsto  P_{\ell,r}^{\ket{\theta}}  $} is continuous by construction.

To finish the proof, we will use this continuity to show that $B^\ell_r$ is compact. In fact, the set $B^\ell_r$ defined in~\cref{eq_setB} can be rewritten as
\begin{align}
    B^\ell_r \,=\,\Big\{\proj{\Psi} \in \St(\cK^{\otimes \ell}) \,:\,\exists\, \ket{\theta} \in \cB(\cK)\ \mathrm{s.t.} \ P_{\ell,r}^{\ket{\theta}}\ket{\Psi} = \ket{\Psi} \Big\}\,.
\end{align}
Consider the product space \smash{$X^\ell := K_\theta \times K_\Psi $} where \smash{$K_\theta = \cB(\cK) $} is the unit sphere in $\cK$ and $K_\Psi$ denotes the set of all pure states $\proj{\Psi}$ in $\St(\cK^{\otimes \ell})$. As both sets $K_\theta$ and $K_\Psi $ are compact (since $\cK$ is finite dimensional), also $X^\ell$ is compact~\cite[Theorem 26.7]{munkres_book}. Furthermore, consider the function
\begin{align}
    F^\ell_r : X^\ell \to \R\,,\quad \big(\ket{\theta},\proj{\Psi} \big) \mapsto \norm{\big(\id - P_{\ell,r}^{\ket{\theta}}\big)\ket{\Psi}}^2_2 = 1 - \bra{\Psi}P_{\ell,r}^{\ket{\theta}}\ket{\Psi}\,, 
\end{align}
which is continuous since the maps \smash{$\ket{\theta}\mapsto P_{\ell,r}^{\ket{\theta}}$} and $(P, \ket{\Psi}) \mapsto P\ket{\Psi}$ are both continuous~\cite[Theorem 19.6]{munkres_book}. Therefore, the set 
\begin{align}
    Z^\ell \,:=\, (F^\ell_r )^{-1}\big(\{0\} \big) \,=\,\Big\{\big(\ket{\theta},\proj{\Psi} \big) \in X^\ell \,:\, P_{\ell,r}^{\ket{\theta}}\ket{\Psi} = \ket{\Psi} \Big\}
\end{align}
is a closed subset $Z^\ell \subseteq X^\ell$, and since $X^\ell$ is compact, $Z^\ell$ is also compact~\cite[Theorem 26.2]{munkres_book}. But since $B^\ell_r = \pi_{\Psi}(Z^\ell)$ where $\pi_{\Psi}$ is the projection of $Z^\ell$ onto the second component (which is also continuous), we find that $B^\ell_r$ is compact~\cite[Theorem 26.5]{munkres_book}.

\end{proof}

\bibliographystyle{arxiv_no_month}
\bibliography{bibliofile}

\end{document}

%% file: iid_HT.tex
\tikzset{meter/.append style={draw, inner sep=10, rectangle, font=\vphantom{A}, minimum width=30, line width=.8, white,
 path picture={\draw[black] ([shift={(.1,.3)}]path picture bounding box.south west) to[bend left=50] ([shift={(-.1,.3)}]path picture bounding box.south east);\draw[black,-latex] ([shift={(0,.1)}]path picture bounding box.south) -- ([shift={(.3,-.1)}]path picture bounding box.north);}}}

\begin{tikzpicture}
\def \x{0.55};
\node[rotate=90,cloud, draw,cloud puffs=10,cloud puff arc=120, aspect=2, inner ysep=1.0em] at (0,0) {};
\draw[Thistle,fill=Thistle] (0,+0.4) circle (0.35);
 \node at (0,0.4) {$\rho^{\otimes n}$};
\draw[BlueGreen,fill=BlueGreen] (0,-0.4) circle (0.35);
 \node at (0,-0.4) {$\sigma^{\otimes n}$}; 
 \node at (0,1.4) {iid source};

\node[meter,scale=0.65] at (3+\x,0) {};
 \draw[thick] (1,0.3) rectangle (3.34+\x,-0.3);
\node at (1.7+0.45,0) {$\{M_n, \id-M_n \}$};
\node at (0.5+0.5*3.34+0.5*\x,0.6) {measurement};

\draw[-latex] (0.3,0.3) -- (1,0.1);
\draw[-latex] (0.3,-0.3) -- (1,-0.1);
\draw (3.35+\x,0.03) -- (3.6+\x,0.03);
\draw (3.35+\x,-0.03) -- (3.6+\x,-0.03);

\node at (4.5+\x,0.4) {guess};
\draw[Thistle,fill=Thistle] (4.0+\x,0) circle (0.25);
 \node at (4.0+\x,0) {$\rho$};
\node at (4.5+\x,0) {or};
\draw[BlueGreen,fill=BlueGreen] (5+\x,0) circle (0.25);
 \node at (5+\x,0) {$\sigma$}; 
  
\end{tikzpicture}

%% file: non_iid_HT.tex
\tikzset{meter/.append style={draw, inner sep=10, rectangle, font=\vphantom{A}, minimum width=30, line width=.8, white,
 path picture={\draw[black] ([shift={(.1,.3)}]path picture bounding box.south west) to[bend left=50] ([shift={(-.1,.3)}]path picture bounding box.south east);\draw[black,-latex] ([shift={(0,.1)}]path picture bounding box.south) -- ([shift={(.3,-.1)}]path picture bounding box.north);}}}

\begin{tikzpicture}
\def \x{0.55};
\node[rotate=90,cloud, draw,cloud puffs=10,cloud puff arc=120, aspect=2, inner ysep=1.0em] at (0,0) {};
\draw[Thistle!50,fill=Thistle!50] (0,+0.4) circle (0.35);
 \node at (0,0.4) {$\rho_{n}$};
\draw[BlueGreen!30,fill=BlueGreen!30] (0,-0.4) circle (0.35);
 \node at (0,-0.4) {$\sigma_{n}$}; 
 \node at (0.5,1.4) {almost-iid source};

\node[meter,scale=0.65] at (3+\x,0) {};
 \draw[thick] (1,0.3) rectangle (3.34+\x,-0.3);
\node at (1.7+0.45,0) {$\{M_n, \id-M_n \}$};
\node at (0.5+0.5*3.34+0.5*\x,0.6) {measurement};

\draw[-latex] (0.3,0.3) -- (1,0.1);
\draw[-latex] (0.3,-0.3) -- (1,-0.1);
\draw (3.35+\x,0.03) -- (3.6+\x,0.03);
\draw (3.35+\x,-0.03) -- (3.6+\x,-0.03);

\node at (4.5+\x,0.4) {guess};
\draw[Thistle,fill=Thistle] (4.0+\x,0) circle (0.25);
 \node at (4.0+\x,0) {$\rho$};
\node at (4.5+\x,0) {or};
\draw[BlueGreen,fill=BlueGreen] (5+\x,0) circle (0.25);
 \node at (5+\x,0) {$\sigma$}; 
  
\end{tikzpicture}


 


  

%% file: iid_HT_ET.tex
\tikzset{meter/.append style={draw, inner sep=10, rectangle, font=\vphantom{A}, minimum width=30, line width=.8, white,
 path picture={\draw[black] ([shift={(.1,.3)}]path picture bounding box.south west) to[bend left=50] ([shift={(-.1,.3)}]path picture bounding box.south east);\draw[black,-latex] ([shift={(0,.1)}]path picture bounding box.south) -- ([shift={(.3,-.1)}]path picture bounding box.north);}}}

\begin{tikzpicture}
\def \x{0.55};
\node[rotate=90,cloud, draw,cloud puffs=10,cloud puff arc=120, aspect=2, inner ysep=1.0em] at (0,0) {};
\draw[Thistle,fill=Thistle] (0,+0.4) circle (0.35);
 \node at (0,0.4) {$\rho_{AB}^{\otimes n}$};
\draw[BlueGreen!30,fill=BlueGreen!30] (0,-0.4) circle (0.35);
 \node at (0,-0.4) {$\mathrm{SEP}$}; 
 \node at (0,1.4) {source};

\node[meter,scale=0.65] at (3+\x,0) {};
 \draw[thick] (1,0.3) rectangle (3.34+\x,-0.3);
\node at (1.7+0.45,0) {$\{M_n, \id-M_n \}$};
\node at (0.5+0.5*3.34+0.5*\x,0.6) {measurement};

\draw[-latex] (0.3,0.3) -- (1,0.1);
\draw[-latex] (0.3,-0.3) -- (1,-0.1);
\draw (3.35+\x,0.03) -- (3.6+\x,0.03);
\draw (3.35+\x,-0.03) -- (3.6+\x,-0.03);

\node at (4.7+\x,0.5) {guess};
\draw[Thistle,fill=Thistle] (4.1+\x,0) circle (0.35);
 \node at (4.1+\x,0) {$\rho_{AB}$};
\node at (4.7+\x,0) {or};
\draw[BlueGreen!30,fill=BlueGreen!30] (5.3+\x,0) circle (0.35);
 \node at (5.3+\x,0) {$\mathrm{SEP}$}; 
  
\end{tikzpicture}


 


  

%% file: non_iid_HT_ET.tex
\tikzset{meter/.append style={draw, inner sep=10, rectangle, font=\vphantom{A}, minimum width=30, line width=.8, white,
 path picture={\draw[black] ([shift={(.1,.3)}]path picture bounding box.south west) to[bend left=50] ([shift={(-.1,.3)}]path picture bounding box.south east);\draw[black,-latex] ([shift={(0,.1)}]path picture bounding box.south) -- ([shift={(.3,-.1)}]path picture bounding box.north);}}}

\begin{tikzpicture}
\def \x{0.55};
\node[rotate=90,cloud, draw,cloud puffs=10,cloud puff arc=120, aspect=2, inner ysep=1.0em] at (0,0) {};
\draw[Thistle!50,fill=Thistle!50] (0,+0.4) circle (0.35);
 \node at (0,0.4) {$\rho_n$};
\draw[BlueGreen!30,fill=BlueGreen!30] (0,-0.4) circle (0.35);
 \node at (0,-0.4) {$\mathrm{SEP}$}; 
 \node at (0,1.4) {source};

\node[meter,scale=0.65] at (3+\x,0) {};
 \draw[thick] (1,0.3) rectangle (3.34+\x,-0.3);
\node at (1.7+0.45,0) {$\{M_n, \id-M_n \}$};
\node at (0.5+0.5*3.34+0.5*\x,0.6) {measurement};

\draw[-latex] (0.3,0.3) -- (1,0.1);
\draw[-latex] (0.3,-0.3) -- (1,-0.1);
\draw (3.35+\x,0.03) -- (3.6+\x,0.03);
\draw (3.35+\x,-0.03) -- (3.6+\x,-0.03);

\node at (4.7+\x,0.5) {guess};
\draw[Thistle,fill=Thistle] (4.1+\x,0) circle (0.35);
 \node at (4.1+\x,0) {$\rho_{AB}$};
\node at (4.7+\x,0) {or};
\draw[BlueGreen!30,fill=BlueGreen!30] (5.3+\x,0) circle (0.35);
 \node at (5.3+\x,0) {$\mathrm{SEP}$}; 
  
\end{tikzpicture}


 


  

%% file: overview_results.tex
\begin{tikzpicture}
\def \xrec{6.2};
\def \yrec{6};

\def \xsrec{5.9};
\def \ysrec{1.5};

\def \xspace{9} 

\def \xs{0.15}; 
\def \ys{0.15}; 

\def \yFontBox{0.29}; 
\def \yFontInBox{0.48}; 

\def \xarrow{0.8}; 

\fill[teal!5,rounded corners] (0,0) rectangle (\xrec,\yrec);
\fill[blue!5,rounded corners] (\xspace,0) rectangle (\xspace+\xrec,\yrec);

\fill[gray!15,rounded corners] (\xs,\ys) rectangle (\xs+\xsrec,\ys+\ysrec);
\node at (\xs+1/2*\xsrec,\ys+\ysrec-\yFontBox) {\small{\textbf{Robust Stein's lemma} (\Cshref{eq_ASL})}}; 
\node at (\xs+1/2*\xsrec,\ys+\ysrec-\yFontBox-\yFontInBox) {\small{$\rho^{\textnormal{almost-iid}}$ vs.~$\sigma^{\textnormal{almost-iid}}$}};
\node at (\xs+1/2*\xsrec,\ys+\ysrec-\yFontBox-2*\yFontInBox) {\small{\cref{thm_ASL}}};

\fill[gray!15,rounded corners] (\xs,\yrec-\ysrec-\ys) rectangle (\xs+\xsrec,\yrec-\ys);
\node at (\xs+1/2*\xsrec,\yrec-\ys-\yFontBox) {\small{\textbf{Stein's lemma} (\Cshref{eq_SL})}}; 
\node at (\xs+1/2*\xsrec,\yrec-\ys-\yFontBox-\yFontInBox) {\small{$\rho^{\textnormal{iid}}$ vs.~$\sigma^{\textnormal{iid}}$}};
\node at (\xs+1/2*\xsrec,\yrec-\ys-\yFontBox-2*\yFontInBox) {\small{\cite{PH91,ogawa00}}};

\fill[gray!15,rounded corners] (\xspace+\xs,\ys) rectangle (\xspace+\xs+\xsrec,\ys+\ysrec);
\node at (\xspace+\xs+1/2*\xsrec,\ys+\ysrec-\yFontBox) {\small{\textbf{Robust gen.~Stein's lemma} (\Cshref{eq_AGSL})}}; 
\node at (\xspace+\xs+1/2*\xsrec,\ys+\ysrec-\yFontBox-\yFontInBox) {\small{$\rho^{\textnormal{almost-iid}}$ vs.~$\mathrm{SEP}$}};
\node at (\xspace+\xs+1/2*\xsrec,\ys+\ysrec-\yFontBox-2*\yFontInBox) {\small{\cref{thm_AGSL}}};

\fill[gray!15,rounded corners] (\xspace+\xs,\yrec-\ysrec-\ys) rectangle (\xspace+\xs+\xsrec,\yrec-\ys);
\node at (\xspace+\xs+1/2*\xsrec,\yrec-\ys-\yFontBox) {\small{\textbf{Generalized Stein's lemma} (\Cshref{eq_generalized_quantum_stein_intro})}}; 
\node at (\xspace+\xs+1/2*\xsrec,\yrec-\ys-\yFontBox-\yFontInBox) {\small{$\rho^{\textnormal{iid}}$ vs.~$\mathrm{SEP}$}};
\node at (\xspace+\xs+1/2*\xsrec,\yrec-\ys-\yFontBox-2*\yFontInBox) {\small{\cite{brandao_Stein_10,haya_stein_25,ludo25}}}; 

\draw[very thick,-latex] (\xarrow,\ys+\ysrec+0.5*\yFontBox) -- (\xarrow,\yrec-\ysrec-\ys-0.5*\yFontBox);
\node[rotate=90] at (\xarrow-0.2,0.5*\yrec-0.15) {\footnotesize{special case}};

\draw[very thick,-latex] (\xspace+\xrec-\xarrow,\ys+\ysrec+0.5*\yFontBox) -- (\xspace+\xrec-\xarrow,\yrec-\ysrec-\ys-0.5*\yFontBox);
\node[rotate=-90] at (\xspace+\xrec-\xarrow+0.2,0.5*\yrec-0.15) {\footnotesize{special case}};

\draw[very thick,-latex] (\xspace-0.2*\xarrow,\yrec-0.5*\ysrec-\ys) -- (\xrec+0.2*\xarrow,\yrec-0.5*\ysrec-\ys);
 \node at (0.5*\xspace+0.5*\xrec,\yrec-0.5*\ysrec-\ys+0.8*\yFontBox) {\footnotesize{special case}};

\draw[very thick,-latex,dashed] (\xspace-0.2*\xarrow,\ys+0.5*\ysrec) -- (\xrec+0.2*\xarrow,\ys+0.5*\ysrec);

\draw[very thick,-latex,dashed] (0.5*\xspace+0.5*\xrec+1.3,0.5*\yrec -0.15) -- (0.5*\xspace+0.5*\xrec+2.5,0.5*\yrec -0.15);

\draw[very thick,-latex,dashed] (0.5*\xspace+0.5*\xrec-1.3,0.5*\yrec -0.15) -- (0.5*\xspace+0.5*\xrec-2.5,0.5*\yrec -0.15);

\node[teal] at (0.45*\xrec,0.5*\yrec) {\small{error}};
\node[teal] at (0.45*\xrec,0.5*\yrec-0.3) {\small{exponent}};
\node[teal] at (0.45*\xrec+1.37,0.5*\yrec-0.15) {\small{$D(\rho \| \sigma)$}};

\node[BlueViolet] at (\xspace+2,0.5*\yrec) {\small{error}};
\node[BlueViolet] at (\xspace+2,0.5*\yrec-0.3) {\small{exponent}};
\node[BlueViolet] at (\xspace+2+1.8,0.5*\yrec-0.15) {\small{$D^{\reg}(\rho \| \mathrm{SEP})$}};

\fill[magenta!10] (0.5*\xspace+0.5*\xrec, 0.5*\yrec -0.15) ellipse (1.2cm and 0.8cm);
\node[magenta] at (0.5*\xspace+0.5*\xrec, 0.5*\yrec +0.2) {\small{Wasserstein}};
\node[magenta] at (0.5*\xspace+0.5*\xrec, 0.5*\yrec -0.15) {\small{continuity}};
\node[magenta] at (0.5*\xspace+0.5*\xrec, 0.5*\yrec -0.5) {\small{\Cshref{prop_continuity_ER_Wasserstein}}};
\end{tikzpicture}

%% file: main.bbl
\begin{thebibliography}{10}

\bibitem{anshu19}
A.~Anshu, M.~Berta, R.~Jain, and M.~Tomamichel.
\newblock {A minimax approach to one-shot entropy inequalities}.
\newblock {\em Journal of Mathematical Physics}, 60(12):122201, 2019.
\newblock
  \texttt{\href{http://dx.doi.org/10.1063/1.5126723}{DOI:\,10.1063/1.5126723}}.

\bibitem{barvinok2002course}
A.~Barvinok.
\newblock {\em A Course in Convexity}.
\newblock Graduate studies in mathematics. American Mathematical Society, 2002.
\newblock
  \texttt{\href{http://dx.doi.org/10.1090/gsm/054}{DOI:\,10.1090/gsm/054}}.

\bibitem{Berta2023gapinproofof}
M.~Berta, F.~G. S.~L. Brand{\~{a}}o, G.~Gour, L.~Lami, M.~B. Plenio, B.~Regula,
  and M.~Tomamichel.
\newblock On a gap in the proof of the generalised quantum {S}tein's lemma and
  its consequences for the reversibility of quantum resources.
\newblock {\em {Quantum}}, 7:1103, 2023.
\newblock
  \texttt{\href{http://dx.doi.org/10.22331/q-2023-09-07-1103}{DOI:\,10.22331/q-2023-09-07-1103}}.

\bibitem{stein_nature_24}
M.~Berta, F.~G. S.~L. Brand{\~a}o, G.~Gour, L.~Lami, M.~B. Plenio, B.~Regula,
  and M.~Tomamichel.
\newblock The tangled state of quantum hypothesis testing.
\newblock {\em Nature Physics}, 20(2):172--175, 2024.
\newblock
  \texttt{\href{http://dx.doi.org/10.1038/s41567-023-02289-9}{DOI:\,10.1038/s41567-023-02289-9}}.

\bibitem{brandao_talk}
F.~G. S.~L. Brand{\~a}o.
\newblock A reversible theory of resources for almost-iid states, 2023.
\newblock Available online:
  \url{https://youtu.be/Um1_7qeA0Uo?si=CvWr6aJ7nGcEFv4v}.
\newblock Talk at the quantum resources conference, Singapore.

\bibitem{brandao_Stein_10}
F.~G. S.~L. Brand{\~a}o and M.~B. Plenio.
\newblock A generalization of quantum {S}tein's lemma.
\newblock {\em Communications in Mathematical Physics}, 295(3):791--828, 2010.
\newblock
  \texttt{\href{http://dx.doi.org/10.1007/s00220-010-1005-z}{DOI:\,10.1007/s00220-010-1005-z}}.

\bibitem{buscemi10}
F.~{Buscemi} and N.~{Datta}.
\newblock The quantum capacity of channels with arbitrarily correlated noise.
\newblock {\em IEEE Transactions on Information Theory}, 56(3):1447--1460,
  2010.
\newblock
  \texttt{\href{http://dx.doi.org/10.1109/TIT.2009.2039166}{DOI:\,10.1109/TIT.2009.2039166}}.

\bibitem{Christandl2017}
M.~Christandl and A.~M{\"u}ller-Hermes.
\newblock Relative entropy bounds on quantum, private and repeater capacities.
\newblock {\em Communications in Mathematical Physics}, 353(2):821--852, 2017.
\newblock
  \texttt{\href{http://dx.doi.org/10.1007/s00220-017-2885-y}{DOI:\,10.1007/s00220-017-2885-y}}.

\bibitem{CW04}
M.~Christandl and A.~Winter.
\newblock “{S}quashed entanglement”: An additive entanglement measure.
\newblock {\em Journal of Mathematical Physics}, 45(3):829--840, 2004.
\newblock
  \texttt{\href{http://dx.doi.org/10.1063/1.1643788}{DOI:\,10.1063/1.1643788}}.

\bibitem{datta09}
N.~Datta.
\newblock Min- and max-relative entropies and a new entanglement monotone.
\newblock {\em IEEE Transactions on Information Theory}, 55(6):2816--2826,
  2009.
\newblock
  \texttt{\href{http://dx.doi.org/10.1109/TIT.2009.2018325}{DOI:\,10.1109/TIT.2009.2018325}}.

\bibitem{PMTL21}
G.~De~Palma, M.~Marvian, D.~Trevisan, and S.~Lloyd.
\newblock The quantum {W}asserstein distance of order 1.
\newblock {\em IEEE Transactions on Information Theory}, pages 1--1, 2021.
\newblock
  \texttt{\href{http://dx.doi.org/10.1109/TIT.2021.3076442}{DOI:\,10.1109/TIT.2021.3076442}}.

\bibitem{PT_23}
G.~De~Palma and D.~Trevisan.
\newblock The {W}asserstein distance of order 1 for quantum spin systems on
  infinite lattices.
\newblock {\em Annales Henri Poincar{\'e}}, 24(12):4237--4282, 2023.
\newblock
  \texttt{\href{http://dx.doi.org/10.1007/s00023-023-01340-y}{DOI:\,10.1007/s00023-023-01340-y}}.

\bibitem{MHR02}
M.~J. Donald, M.~Horodecki, and O.~Rudolph.
\newblock The uniqueness theorem for entanglement measures.
\newblock {\em Journal of Mathematical Physics}, 43(9):4252--4272, 2002.
\newblock
  \texttt{\href{http://dx.doi.org/10.1063/1.1495917}{DOI:\,10.1063/1.1495917}}.

\bibitem{dupuis12}
F.~Dupuis, L.~Kraemer, P.~Faist, J.~M. Renes, and R.~Renner.
\newblock {Generalized Entropies}.
\newblock In {\em Proc. XVIIth International Congress on Mathematical Physics},
  pages 134--153, Aalborg, Denmark, 2012.
\newblock
  \texttt{\href{http://dx.doi.org/10.1142/9789814449243{\_}0008}{DOI:\,10.1142/9789814449243{\_}0008}}.

\bibitem{FFF25}
K.~Fang, H.~Fawzi, and O.~Fawzi.
\newblock Generalized quantum asymptotic equipartition, 2025.
\newblock Available online: \url{https://arxiv.org/abs/2411.04035}.

\bibitem{fekete23}
M.~Fekete.
\newblock {\"U}ber die {V}erteilung der {W}urzeln bei gewissen algebraischen
  {G}leichungen mit ganzzahligen {K}oeffizienten.
\newblock {\em Mathematische Zeitschrift}, 17(1):228--249, 1923.

\bibitem{Giacomo_private}
F.~Girardi, G.~D. Palma, and L.~Lami.
\newblock New approaches to almost i.i.d.~information theory, 2026.
\newblock Available online: \url{https://arxiv.org/abs/2605.15114}.

\bibitem{hastings09}
M.~B. Hastings.
\newblock Superadditivity of communication capacity using entangled inputs.
\newblock {\em Nature Physics}, 5(4):255--257, 2009.
\newblock
  \texttt{\href{http://dx.doi.org/10.1038/nphys1224}{DOI:\,10.1038/nphys1224}}.

\bibitem{haya_stein_25}
M.~Hayashi and H.~Yamasaki.
\newblock The generalized quantum {S}tein's lemma and the second law of quantum
  resource theories.
\newblock {\em Nature Physics}, 21(12):1988--1993, 2025.
\newblock
  \texttt{\href{http://dx.doi.org/10.1038/s41567-025-03047-9}{DOI:\,10.1038/s41567-025-03047-9}}.

\bibitem{PH91}
F.~Hiai and D.~Petz.
\newblock The proper formula for relative entropy and its asymptotics in
  quantum probability.
\newblock {\em Communications in Mathematical Physics}, 143(1):99--114, 1991.
\newblock
  \texttt{\href{http://dx.doi.org/10.1007/BF02100287}{DOI:\,10.1007/BF02100287}}.

\bibitem{analysis_banach}
T.~Hyt\"onen, J.~van Neerven, M.~Veraar, and L.~Weis.
\newblock {\em Analysis in Banach Spaces}.
\newblock Volume I: Martingales and Littlewood-Paley Theory. Springer, 2016.
\newblock
  \texttt{\href{http://dx.doi.org/10.1007/978-3-319-48520-1}{DOI:\,10.1007/978-3-319-48520-1}}.

\bibitem{Ludo25_2}
L.~Lami.
\newblock A doubly composite {C}hernoff-{S}tein lemma and its applications,
  2025.
\newblock Available online: \url{https://arxiv.org/abs/2510.06342}.

\bibitem{ludo25}
L.~Lami.
\newblock A solution of the generalized quantum {S}tein’s lemma.
\newblock {\em IEEE Transactions on Information Theory}, 71(6):4454--4484,
  2025.
\newblock
  \texttt{\href{http://dx.doi.org/10.1109/TIT.2025.3543610}{DOI:\,10.1109/TIT.2025.3543610}}.

\bibitem{LRT26}
L.~Lami, B.~Regula, and R.~Takagi.
\newblock Universal quantum resource distillation via composite generalised
  quantum {S}tein's lemma, 2026.
\newblock Available online: \url{https://arxiv.org/abs/2605.15174}.

\bibitem{MSR26}
G.~Mazzola, D.~Sutter, and R.~Renner.
\newblock Almost-iid information theory, 2026.
\newblock Available online: \url{https://arxiv.org/abs/2603.15792}.

\bibitem{MLDSFT13}
M.~M\"{u}ller-Lennert, F.~Dupuis, O.~Szehr, S.~Fehr, and M.~Tomamichel.
\newblock On quantum {R}\'enyi entropies: A new generalization and some
  properties.
\newblock {\em Journal of Mathematical Physics}, 54(12), 2013.
\newblock
  \texttt{\href{http://dx.doi.org/http://dx.doi.org/10.1063/1.4838856}{DOI:\,http://dx.doi.org/10.1063/1.4838856}}.

\bibitem{munkres_book}
J.~R. Munkres.
\newblock {\em Topology}.
\newblock Prentice Hall, second edition, 2000.

\bibitem{ogawa00}
T.~{Ogawa} and H.~{Nagaoka}.
\newblock Strong converse and {S}tein's lemma in quantum hypothesis testing.
\newblock {\em IEEE Transactions on Information Theory}, 46(7):2428--2433,
  2000.
\newblock
  \texttt{\href{http://dx.doi.org/10.1109/18.887855}{DOI:\,10.1109/18.887855}}.

\bibitem{piani09}
M.~Piani.
\newblock Relative entropy of entanglement and restricted measurements.
\newblock {\em Phys. Rev. Lett.}, 103:160504, 2009.
\newblock
  \texttt{\href{http://dx.doi.org/10.1103/PhysRevLett.103.160504}{DOI:\,10.1103/PhysRevLett.103.160504}}.

\bibitem{BLD26}
B.~Regula, L.~Lami, and N.~Datta.
\newblock Tight relations and equivalences between smooth relative entropies.
\newblock {\em IEEE Transactions on Information Theory}, 72(5):3051--3073,
  2026.
\newblock
  \texttt{\href{http://dx.doi.org/10.1109/TIT.2026.3661711}{DOI:\,10.1109/TIT.2026.3661711}}.

\bibitem{renner_phd}
R.~Renner.
\newblock Security of quantum key distribution.
\newblock {\em PhD thesis, ETH Zurich}, 2005.
\newblock available at
  \texttt{arXiv:\href{http://arxiv.org/abs/quant-ph/0512258}{quant-ph/0512258}}.

\bibitem{renner07}
R.~Renner.
\newblock Symmetry of large physical systems implies independence of
  subsystems.
\newblock {\em Nature Physics}, 3(9):pp. 645--649, 2007.
\newblock Available online:
  \url{http://www.nature.com/nphys/journal/v3/n9/suppinfo/nphys684_S1.html}.

\bibitem{Sutter_book}
D.~Sutter.
\newblock {\em Approximate Quantum Markov Chains}, pages 75--100.
\newblock Springer International Publishing, Cham, 2018.
\newblock
  \texttt{\href{http://dx.doi.org/10.1007/978-3-319-78732-9\_5}{DOI:\,10.1007/978-3-319-78732-9\_5}}.

\bibitem{marco_book}
M.~Tomamichel.
\newblock {\em Quantum Information Processing with Finite Resources}, volume~5
  of {\em SpringerBriefs in Mathematical Physics}.
\newblock Springer, 2015.
\newblock
  \texttt{\href{http://dx.doi.org/10.1007/978-3-319-21891-5}{DOI:\,10.1007/978-3-319-21891-5}}.
\newblock See \url{https://arxiv.org/abs/1504.00233} for the precise
  references.

\bibitem{Uhl76}
A.~Uhlmann.
\newblock The ``transition probability'' in the state space of a {*}-algebra.
\newblock {\em Reports on Mathematical Physics}, 9(2):273 -- 279, 1976.
\newblock
  \texttt{\href{http://dx.doi.org/10.1016/0034-4877(76)90060-4}{DOI:\,10.1016/0034-4877(76)90060-4}}.

\bibitem{VPRK97}
V.~Vedral, M.~B. Plenio, M.~A. Rippin, and P.~L. Knight.
\newblock Quantifying entanglement.
\newblock {\em Phys. Rev. Lett.}, 78:2275--2279, 1997.
\newblock
  \texttt{\href{http://dx.doi.org/10.1103/PhysRevLett.78.2275}{DOI:\,10.1103/PhysRevLett.78.2275}}.

\bibitem{VW01}
K.~G.~H. Vollbrecht and R.~F. Werner.
\newblock Entanglement measures under symmetry.
\newblock {\em Phys. Rev. A}, 64:062307, 2001.
\newblock
  \texttt{\href{http://dx.doi.org/10.1103/PhysRevA.64.062307}{DOI:\,10.1103/PhysRevA.64.062307}}.

\bibitem{WR12}
L.~Wang and R.~Renner.
\newblock One-shot classical-quantum capacity and hypothesis testing.
\newblock {\em Phys. Rev. Lett.}, 108:200501, 2012.
\newblock
  \texttt{\href{http://dx.doi.org/10.1103/PhysRevLett.108.200501}{DOI:\,10.1103/PhysRevLett.108.200501}}.

\bibitem{wilde_book}
M.~M. Wilde.
\newblock {\em Quantum Information Theory}.
\newblock Cambridge University Press, 2013.
\newblock
  \texttt{\href{http://dx.doi.org/10.1017/9781316809976}{DOI:\,10.1017/9781316809976}}.

\bibitem{WWY14}
M.~M. Wilde, A.~Winter, and D.~Yang.
\newblock Strong converse for the classical capacity of entanglement-breaking
  and {H}adamard channels via a sandwiched {R}\'enyi relative entropy.
\newblock {\em Communications in Mathematical Physics}, 331(2):593--622, 2014.
\newblock
  \texttt{\href{http://dx.doi.org/10.1007/s00220-014-2122-x}{DOI:\,10.1007/s00220-014-2122-x}}.

\bibitem{winter99}
A.~Winter.
\newblock Coding theorem and strong converse for quantum channels.
\newblock {\em IEEE Transactions on Information Theory}, 45(7):2481--2485,
  1999.
\newblock
  \texttt{\href{http://dx.doi.org/10.1109/18.796385}{DOI:\,10.1109/18.796385}}.

\bibitem{win16}
A.~Winter.
\newblock Tight uniform continuity bounds for quantum entropies: Conditional
  entropy, relative entropy distance and energy constraints.
\newblock {\em Communications in Mathematical Physics}, 347(1):291--313, 2016.
\newblock
  \texttt{\href{http://dx.doi.org/10.1007/s00220-016-2609-8}{DOI:\,10.1007/s00220-016-2609-8}}.

\bibitem{wolf_notes}
M.~M. Wolf.
\newblock Quantum channels \& operations: Guided tour, 2012.
\newblock Lecture notes available at
  \url{https://www-m5.ma.tum.de/foswiki/pub/M5/Allgemeines/MichaelWolf/QChannelLecture.pdf}.

\end{thebibliography}
